\documentclass[letterpaper,twocolumn,10pt]{article}
\usepackage{usenix,epsfig,endnotes}

\usepackage{amsthm}
\newtheorem{Definition}{Definition}
\newtheorem{Lemma}{Lemma}
\newtheorem{Theorem}{Theorem}
\newtheorem{Problem}{Problem}
\newtheorem{Observation}{Observation}
\usepackage[linesnumbered, ruled,vlined]{algorithm2e}

\usepackage{enumerate}
\usepackage{graphicx}
\usepackage{booktabs}

\usepackage{hyperref}

\usepackage{slashbox}

\usepackage{multirow}

\newcommand{\ignore}[1]{{}}

\usepackage[subrefformat=parens,labelformat=parens]{subfig}
\usepackage{caption}
\usepackage{array}
\newcommand{\PreserveBackslash}[1]{\let\temp=\\#1\let\\=\temp}
\newcolumntype{C}[1]{>{\PreserveBackslash\centering}m{#1}}
\newcolumntype{R}[1]{>{\PreserveBackslash\raggedleft}m{#1}}
\newcolumntype{L}[1]{>{\PreserveBackslash\raggedright}m{#1}}

\usepackage{eucal}

\begin{document}

\date{}

\title{\Large \bf Two Dimensional Router: Design and Implementation}

\author{
{\rm Shu Yang, Laizhong Cui}\\
Shenzhen University
\and
{\rm Xinhao Deng, Qi Li, Mingwei Xu, Jianping Wu}\\
Tsinghua University
\and
{\rm Yulei Wu}\\
University of Exeter
\and
{\rm Dan Wang}\\
Hong Kong Polytechnic University
} 

\maketitle

\thispagestyle{empty}

\subsection*{Abstract}
Higher dimensional classification has attracted more attentions with increasing demands for more
flexible services in the Internet. In this paper, we present the design and implementation of a
two dimensional router (TwoD router), that makes forwarding decisions based on both destination
and source addresses. This TwoD router is also a key element in our current effort towards two
dimensional IP routing.

With one more dimension, the forwarding table will grow explosively given a straightforward
implementation. As a result, it is impossible to fit the forwarding table to the current TCAM,
which is the de facto standard despite its limited capacity. To solve the explosion problem,
we propose a forwarding table structure with a novel separation of TCAM and SRAM. As such, we
move the redundancies in expensive TCAM to cheaper SRAM, while the lookup speed is comparable
with conventional routers. We also design the incremental update algorithms that minimize the
number of accesses to memory. We evaluate our design with a real implementation on a commercial
router, Bit-Engine 12004, with real data sets. Our design does not need new devices, which is
favorable for adoption. The results also show that the performance of our TwoD router is promising.

\section{Introduction}
To provide reachability services to the Internet users, conventional Internet routers classify packets
based only on destination address,  Although one dimensional routers are adequate for destination-based
routing, there are increasing demands for higher dimensional routers \cite{Kim08}, for security, traffic engineering,
quality of service, etc. Among the higher dimensional routers, two dimensional routers (TwoD routers), that classify
packets based on both destination and source addresses, have gained considerable attentions
\cite{Baboescu06}\cite{Suri03}\cite{Lu05}, due to the important semantics of destination and source addresses
\cite{network-algorithmics}. For example, TwoD routers can easily express the policies between host and host,
or network and network.

China Education and Research Network 2 (CERNET2), the largest naive IPv6 network around the world, is now deploying
Two Dimensional-IP (TwoD-IP) routing \cite{TwoD-IP-Draft}. More specifically, the routing decisions will be based not
only on destination address, but also on the source address. Such extension provides rooms to solve problems of the
past and foster innovations in the future. TwoD router is a key element in TwoD-IP routing.

There has been many research works on TwoD routers. Most of them focus on software-based solutions
\cite{Wang09}\cite{Vamanan10}\cite{Baboescu05}, however, software-based solutions need
many accesses to memory, and cause non-deterministic lookup time.
The problem gets worse after deploying IPv6, where more bits should be matched. Hardware-based,
especially TCAM-based solutions are the de facto standard for core routers, due to their
constant lookup time and high speeds. Despite its high speeds, TCAM is limited by its low
capacity, large power consumption and high cost \cite{meiners2010}. The largest TCAM chip available currently
can only accommodate 1 million IPv4 prefixes \cite{Meiners11}.

TCAM resources are further limited in TwoD routers. Two dimensional classifiers widely adopt the
traditional Cisco Access Control List (ACL) structure (we call it ACL-like structure thereafter), e.g., CERNET2
is using this structure. In Figure \ref{tab-acl-like-table}, we show a typical table within
ACL-like structure, where destination and source prefixes (having 4 bits for brevity) are concatenated
as an entry in TCAM. For example, receiving a packet with destination address of 1011 and source address
of 1111, router will forward the packet to 1.0.0.2, after matching destination prefix 101* and source prefix 11**
according to the longest match first (LMF) rule. This `fat' TCAM structure provides fast lookup speeds, however,
this structure greatly increases the TCAM resources in TwoD routers, due to 1) it doubles the
width of a TCAM entry, e.g., 288 bits (typical TCAM width) are needed within
IPv6; 2) in the worst case, the number of TCAM entries can be $O(N\times M)$, where $N$ and $M$ are the
space of destination and source addresses. The ACL-like structure works well within a few entries, however,
it becomes inefficient when the number of entries increases. If TwoD-IP routing is deployed, the number of entries
will predictably increase more rapidly, e.g., CERNET2 wants to carry out policy routing between about 6,000 destination
prefixes and 100 source prefixes, resulting in 600,000 entries in TCAM.

\begin{table}[h!]
\scriptsize{
  \centering
  \begin{tabular}{cc|c}
    \toprule
    Destination prefix  &  Source prefix  & Action\\
    \midrule
    111*                &  111*           & Forward to 1.0.0.0     \\
    111*                &  100*           & Forward to 1.0.0.1     \\
    100*                &  111*           & Forward to 1.0.0.2     \\
    101*                &  11**           & Forward to 1.0.0.2     \\
    10**                &  11**           & Forward to 1.0.0.3     \\
    \bottomrule
  \end{tabular}
  \caption{\footnotesize{Table with ACL-like structure in TCAM}}
  \label{tab-acl-like-table}
}
\end{table}

In this paper, to relieve the contradictions, we put forward a new forwarding table structure called FIST (\textbf{FI}B
\textbf{S}tructure for \textbf{T}woD-IP). The key idea of FIST is to store destination and source prefixes in two
separate TCAM tables, and store other information in SRAM, which is much cheaper and less power consumptive than TCAM.
In Figure \ref{tab-acl-like-table}, we need to store destination prefixes 111*, 100*, 101* and 10**
in one TCAM table, and source prefixes 111*, 100*, 11** in another TCAM table. Through moving the redundancies from TCAM
to SRAM, we can reduce the TCAM storage space, because 1) TCAM width can be reduced to be one half, e.g., 144 bits are
enough within IPv6; 2) reducing the number of entries in TCAM, i.e., each prefix appears only once. In the worst
case, there are $N+M$ TCAM entries. Trivial FIST may increase the SRAM storage space, thus we
develop a set of techniques for better scalability. We show that the redundancies in SRAM can be largely removed,
due to the flexibility in SRAM.

Within FIST, each destination prefix points to a row, each source prefix
points to a column, and they both together point to a two dimensional array element (cell) in SRAM, through
which we can compute the action (or next hop) information. When a packet arrives, we can match its destination and
source in parallel in TCAM, and find the next hop information in SRAM. The lookup process can be pipelined, and the
lookup time is comparable with current Internet routers.

However, within FIST, there may exist confliction, i.e., matching a wrong prefix, after removing the binding
relation between destination and source prefixes in TCAM. For example, if a packet with destination address of 1011
and source address of 1111 arrives, destination prefix 101* and source prefix 111* will be matched by applying LMF
rule in each separate TCAM table, however, there does not exist any entry with destination prefix 101* and source
prefix 111*. To resolve such confliction, we pre-compute the right actions for all conflicted cases.
Such pre-computation guarantees the correctness, but it becomes impractical when updates happen frequently,
because it needs re-computation for all conflicted cases, and causes large number of accesses to SRAM once update
happens. To support incremental updates, we propose a new data structure called colored tree, through which we can
minimize the computation cost and number of accesses to memory.

We implement the FIST on a commercial router, Bit-Engine 12004. Through redesigning the hardware logic, we do not
need new devices. We carry out comprehensive evaluations with the real implementation, using the real topology,
FIB, prefix and traffic data from CERNET2. The results show that FIST-based TwoD router can achieve linecard speeds,
save TCAM and SRAM storage space, and bring acceptable update burden.

\section{Overview of TwoD Router Design}

We want the performance (i.e., packet processing) of the TwoD Router to be comparable
with the current Internet routers. We choose TCAM as our base line design as TCAM is the
key factor for the fast speed of the current routers.

The immediate change that TwoD-IP routing brings to the picture is the forwarding table size.
More specifically, the Forwarding Information Base (FIB) will tremendously increase. Note that
a first thought might think that the routing table only doubles. This is not true, as for each
destination address, it corresponds to different source address. A straightforward implementation,
i.e., ACL-like structure, means the FIB table changes from \{destination\} $\rightarrow$ \{action\}
to \{(destination, source)\} $\rightarrow$ \{action\}. This increases the FIB size by an order and
a practical consequence is that TCAM cannot hold entries of such scale. Current TCAM storage
is 1 million and current destination prefix number is 400,000 \cite{bgp-routing-table-size}. If
TwoD-IP is implemented by a straightforward approach, even with 100 source prefixes, it is already
far beyond the TCAM storage.

We solve this problem by proposing a novel forwarding table structure FIST (see Fig. \ref{fig-storage-structure-twod-ip}).
The key of FIST is a novel separation of TCAM and SRAM. TCAM contributes to fast lookup and SRAM contributes
to a larger memory space. Overall, FIST consumes $O(N+M)$ TCAM storage space.

Another difficulty is the update action. In principle, an update of a destination prefix in the
TwoD router may incur an update for each source prefix associated with this destination prefix
and vice versa. This indicates that the update of a single entry in TwoD router is, given a
straightforward design, the same as updating a full table of the current router.

Suppose that there are 10,000 source prefixes, and 500 updates on destination prefixes per second.
In the worst case, there are 5,000,000 updates on SRAM per second, which almost exceeds the
speed of hardware (in BitWay 12004, linecards work at 100MHz, and linecards need
20 clock cycles for a read/write operation).

We try every aspect to reduce the update complexity. We formulate an optimal transformation problem
where we want to minimize the total number of read/write for each update. To solve this problem we
propose a colored tree to organize the entries and we prove that we can minimize the computation
complexity and number of accesses to memory during update actions.

In the following paper, Section \ref{sec-fist-structure} introduces the FIST structure and we prove
its correctness during packet forwarding. We further present the lookup process on FIST. We
discuss the incremental update action in Section \ref{sec-update}. In Section \ref{sec-practical-consideration},
we take some practical issues into consideration and improves the trivial FIST structure.
Section \ref{sec-implementation} presents the implementation of FIST on a commercial router.
Section \ref{sec-evaluate} provides evaluation details and results. In Section \ref{sec-discuss}
and \ref{sec-related-work}, we discuss the scalability of FIST and introduce the related works.
Finally, we present our conclusions in Section \ref{sec-conclusion}.

\section{FIST Structure and Lookup}\label{sec-fist-structure}

\subsection{The TwoD matching rule}
We first present the definition of the forwarding rules that is used in
two dimensional routing. Let $d$ and $s$ denote the destination and source
addresses, $p_d$ and $p_s$ denote the destination and source prefixes.
Let $a$ denote an action, more specifically, the next hop. The storage
structure should have entries of 3-tuple $(p_d, p_s, a)$.

\begin{Definition} \label{def-match} TwoD matching rule:
Assume a packet with $s$ and $d$ arrives at a router. The destination address
$d$ should first match $p_d$ according to the LMF rule. The source address $s$ should
then match $p_s$ according to the LMF rule among all the 3-tuple given that $p_d$ is
matched. The packet is then forwarded to next hop $a$.
\end{Definition}

Our rule is defined based on the following principles: 1) Avoid confliction:
it has been shown \cite{Lu05} that if matching the source and the destination
address with the same priority, the LMF rule cannot decide the priority.
Even using the first-matching-rule-in-table tie breaker may result in loops and
resolving the confliction is NP-hard. 2) Compatibility: Matching destination
prefixes first emphasizes on connectivity and is compatible with previous
destination-based architecture. More specifically, if no source prefix is
involved, our rule naturally regresses to traditional forwarding rules.
Note that our router design is symmetric if source prefix is
matched first.

\begin{figure}[ht]
  \includegraphics[width=3.3in]{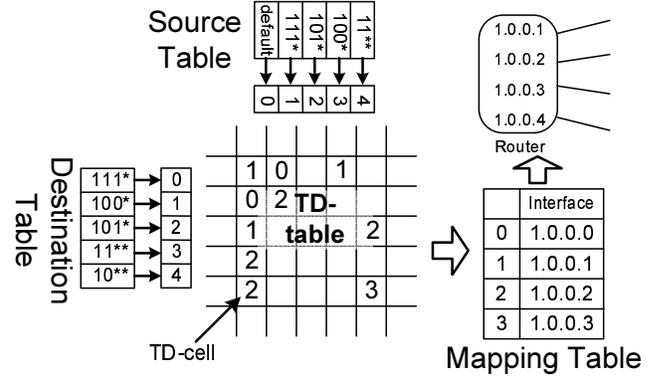}
  \centering
  \caption{\footnotesize{FIST: A forwarding table structure for TwoD-IP}}\label{fig-storage-structure-twod-ip}
\end{figure}

\subsection{FIST Design Details}

\subsubsection{FIST basics}

The new structure FIST is made up of two tables stored in TCAMs and two tables stored in SRAM (see Fig.
\ref{fig-storage-structure-twod-ip}). One table in TCAM stores the destination prefixes (we call it
{\em destination table} thereafter), and the other table in TCAM stores the source
prefixes (we call it {\em source table} thereafter). One table in SRAM is a two dimensional
table that stores the indexed next hop of each rule in TwoD-IP (we call it {\em TD-table} thereafter)
and we call each cell in the array {\em TD-cell} (or in short {\em cell} if no ambiguity).
Another table in SRAM stores the mapping relation of index values and next hops (we call it
{\em mapping-table} thereafter).

For each rule $(p_d, p_s, a)$, $p_d$ is stored in the destination table, $p_s$ is stored in the
source table. For the $(p_d, p_s)$ cell in the TD-table, there stores an index value. From this
index value, $a$ is stored in the corresponding position of mapping table. We store the index
value rather than the next hop $a$ in the TD-table, because the next hop information is much longer.

As an example, in Fig. \ref{fig-storage-structure-twod-ip}, for $(100*,111*,1.0.0.2)$, $100*$ is stored
in the destination table and is associated with the $1^{st}$ row; and $111*$ is stored in the source
table and is associated with the $1^{st}$ column. In the TD-table, the cell $(100*, 111*)$ that
corresponding to $1^{st}$ column and $1^{st}$ row has index value 2. In the mapping table, the next
hop that is related with index value 2 is $1.0.0.2$.

\ignore{
To provide better connectivity, each destination prefix is associated with one or more default next hops. If no source
prefix matches the source address of a packet, then routers will forward the packet to the default next hop associated
with the matched destination prefix. The default next hop can be seen as a string composed of wildcards (**** in
Fig. \ref{fig-storage-structure-twod-ip}). Thus, for any arrived packet, there will be at least one source prefix that
matches its source address.
}

\begin{Theorem}\label{theorem-storage}
The TCAM storage space of FIST is $O(N+M)$ bits. The SRAM storage space of FIST is $O(N\times M \times log(P))$ bits, where $P$ is the size of the mapping table.
\end{Theorem}
\begin{proof}
Because the destination table has $N$ entries, and the source table has $M$ entries, TCAM space is $O(N+M)$ bits. Mapping-table stores the mapping relations between the index and the corresponding next hop interface. They indeed have an upper on the size because each router has a bound on the number of next hop interfaces. Let $P$ denote the number of interfaces of a router and $W$ represent the size of the entries associated with the interfaces. Then the size of mapping-table is less than $PW$. The size of mapping-table can be treated as a constant compared with the SRAM storage space. Therefore, we mainly consider the TD-table size in calculating the SRAM storage space. TD-table dominates the space of SRAM, and has $O(N\times M)$ cells of $log(P)$ bits.
\end{proof}

From Theorem \ref{theorem-storage}, we can see that FIST move the `multiplication' to SRAM, rather than eliminate
it. Such movement is worthwhile considering the following facts: 1) Capacity of TCAM is much smaller
than that of SRAM; 2) TCAM is 10-100 times more expensive than SRAM; 3) TCAM consumes $>100$
times more power than SRAM \cite{liu10}\cite{Chiba10}\cite{router-fib}. Besides, SRAM is more flexible
than TCAM, thus reducing redundancies is more easily.

\subsubsection{TD-cell Saturation}

For the example in Fig. \ref{fig-storage-structure-twod-ip}, if a packet with destination address 1011
and source address 1111 arrives at the router, rule $(101*,11**,1.0.0.2)$ should be matched. This is because according
to LMF rule, the destination prefix 101* should be first matched. There are two rules (including the default rule)
associated with the destination prefix 101*. Consequently, source prefix 11** will be matched.
With the new structure, destination prefix 101* will be matched and source prefix 111* will be matched.
However, the cell $(101*,111*)$ ($2^{nd}$ row and $1^{st}$ column) in TD-Table does
not have any index value.
Intrinsically, consider a packet that should match destination and source prefix pairs $(p_d, p_s)$. If there exists a
source prefix $p_s'$ that is longer than $p_s$, cell $(p_d, p_s')$ rather than $(p_d, p_s)$ will be matched.

To address the problem, we pre-compute and fill the conflicted cells, e.g., $(101*, 111*)$,
with appropriate index value. The algorithm is as follows.

\begin{algorithm}[!h]
\footnotesize{
\Begin{
\ForEach{$p_d, p_s$}{
\If{$\not\exists (p_d, p_s, a)\in \mathcal{R}$}
{
$\mathcal{S}=\{(\bar{p_s},\bar{p_d},\bar{a})\in \mathcal{R}| \bar{p_d}=p_d\}$\label{line-s}\;
$\mathcal{S}' = \{(\tilde{p_s},\tilde{p_d},\tilde{a})\in \mathcal{S}| \tilde{p_s}$ is a prefix of $p_s\}$\label{line-s1}\;
Find $(\hat{p_s},\hat{p_d},\hat{a})\in \mathcal{S}', \forall (p_d', p_s', a')\in \mathcal{S}', p_s'$ is a prefix of $\hat{p_s}$\;\label{line-longest}
Fill the cell $(p_d, p_s)$ with $\hat{a}$.
}
}

}
\caption{\footnotesize{TD-Saturation($\mathcal{R}$)}}
}
\label{<alg-td-saturation>}
\end{algorithm}

We show the TD-table after filling up all the conflicted cells in Figure \ref{fig-twod-array-all-cells}.

\begin{Theorem}
FIST (with \textup{TD-Saturation}()) correctly handle the rule defined in Definition \ref{def-match}.
\end{Theorem}
\begin{proof}
When a packet arrives, and matches $(p_d, p_s)$ according to FIST. If $\exists (p_d, p_s, a)\in \mathcal{R}$.
Then this cell stores the index value of $a$, which is the right one.

Else according to Algorithm TD-Saturation(), $\mathcal{S}$ contains all rules given $p_d$ is matched. In Line 5,
$\tilde{p_s}$ is a prefix of $p_s$, thus the packet also match the rule $(\tilde{p_s},\tilde{p_d},\tilde{a})\in \mathcal{S}'$.
In line 6, because there does not exist $(p_s',p_d',a')\in\mathcal{S}'$ where $p_s'$ is longer than
$\hat{p_s}$, $\hat{p_s}$ is the longest match among all the rules given $p_d$ is matched. So $(p_d, p_s)$ should be
set to be $\hat{a}$ according to Definition \ref{def-match}.
\end{proof}

\subsubsection{A Non-Homogeneous FIST Structure}
\label{sec-improvements-storage}
We expect that in practice, many destination prefixes only have default next hops.
It is thus wasteful to leave a row for the TD-table. To become more compatible to the current router
structure and further reduce the SRAM space, we divide the forwarding table into two parts. In the
first part each prefix points to a row in TD-table, and in the second part each prefix points
directly to an index value. For example, in Fig. \ref{fig-storage-structure-twod-ip}, destination prefix
11** does not need any specific source prefix, thus it is stored in the second part.

In our implementation, we logically divide the table into two parts by using a indicator bit
to separate them. We illustrate more details in Section \ref{sec-implementation}.

\subsection{FIST Lookup}
\begin{figure}[!h]
\centering
  \includegraphics[width=3.4in]{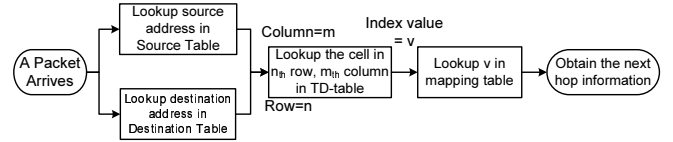}
  \centering
  \caption{\footnotesize{Lookup action in FIST}}
  \label{fig-lookup-action}
\end{figure}

The lookup action $lookup(d, s)$ is shown in Fig. \ref{fig-lookup-action}. When a packet arrives, the router first
extracts the source address $s$ and destination address $d$. Using the LMF rule, the router finds the matched source
and destination prefixes in both source and destination tables that reside in TCAMs. According to the matched entry,
the source table will output a column address and the destination table will output a row address. Combined with the row and
column addresses, the router can find a cell in the TD-table, and return an index value. Using the index value, the router
looks up the mapping table, and returns the next hop information.

\begin{Theorem}
The look up speed of FIST is one TCAM clock cycle plus three SRAM clock cycles.
\end{Theorem}
\begin{proof}
Source and destination tables can be accessed in parallel. Thus one clock cycle of TCAM is enough. Getting the
row and column address cost one SRAM clock cycle,
Then the router will access TD-table, and mapping table, each cost one SRAM clock cycle.
\end{proof}

As a comparison, the conventional destination-based routing usually
stores destination prefixes in one TCAM, and accesses both
TCAM and SRAM for one time during a lookup process.
Note that the SRAM clock cycle is much smaller than TCAM cycle \cite{Kim09}, and
the bottleneck of a router is normally during delivering
packets through the FIFO, thus two more accesses in SRAM will not
have a significant impact on throughput.

To minimize the additional impact on throughput, we develop a pipeline
lookup process (see model in Fig. \ref{fig-pipeline}). When a packet arrives,
the router first extracts the source and destination addresses, and hands them
to the search engine. The router then looks up the source and destination address
in parallel from the source and destination tables. Note that we can
perform such parallel processing because we have saturated the TD-table.
After the router obtains the SRAM addresses that point to the row and column values,
the SRAM addresses are passed to a FIFO buffer, which resolves the un-matching
clock-rates between TCAM and SRAM. Using the SRAM addresses, router looks up the
SRAM that is used in conjunction with TCAM, to get the row and column. Then router
makes use of the row and column values to lookup the TD-table, and obtains the index
value, which is then used to lookup the mapping table. Finally, router looks up the
mapping table and obtains the next hop information.

\begin{figure}[!ht]
\centering
  \includegraphics[width=3.2in]{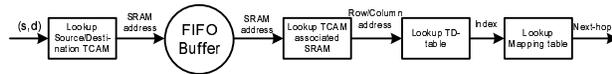}
  \centering
  \caption{\footnotesize{Lookup process pipeline of FIST}}
  \label{fig-pipeline}
\end{figure}

Pipelining itself is not new and almost all routers implement it today.
Using the pipeline, the lookup speed of FIST can achieve one packet per TCAM clock rate.

\begin{Observation}
All implemented with pipelining, the lookup process of the FIST routers is
the same as conventional routers.
\end{Observation}

\section{Forwarding Table Compression}\label{sec-compression}
Let $P_s$ be the set of source prefixes, $P_d$ be the set of destination prefixes.
Let $f_r(p_d), p_d\in P_d$ (or $f_c(p_s), p_s\in P_s$) be a mapping function that maps
a destination (or source) prefix $p_d$ (or $p_s$) to the $f_r(p_d)$ (or $f_c(p_s)$) row
(or column). Let $TD(x, y)$ denote the cell in the $x_{th}$ row and $y_{th}$ column.
We use the 5-tuple $\{P_d, P_s, f_r(\cdot), f_c(\cdot), TD(\cdot, \cdot)\}$ to denote
a forwarding table.

\begin{Definition}
$\{P_d', P_s', f_r'(\cdot), f_c'(\cdot), TD'(\cdot, \cdot)\}$ is \textbf{equivalent} to
$\{P_d, P_s, f_r(\cdot), f_c(\cdot), TD(\cdot, \cdot)\}$,
for any source address $s$ and destination address $d$, $s$ matches $p_s$ in $P_s$ and $p_s'$ in $P_s'$,
$p_d$ in $P_d$ and $p_d'$ in $P_d'$ according to LMF rule, $TD(f_r(p_d),f_c(p_s))=TD'(f_r'(p_d'),f_c'(p_s'))$
is satisfied.
\end{Definition}

For a given forwarding table, our objective is to find an equivalent forwarding table, which
occupies minimum storage space, including both TCAM and SRAM.

\subsection{Compression in TCAM Space}

We first compress the storage space in TCAM, including destination and source tables. The size of
destination and source tables can be measured by the number of destination and source prefixes in
them.

\begin{Problem}
\textbf{Optimal TCAM Compression:}
For $\{P_d, P_s, f_r(\cdot), f_c(\cdot), TD(\cdot, \cdot)\}$, find an equivalent forwarding table
$\{P_d', P_s', f_r'(\cdot), f_c'(\cdot), TD'(\cdot, \cdot)\}$ such that the storage space in TCAM,
i.e., $|P_s'|+|P_d'|$ is minimized.
\end{Problem}

We develop algorithm \emph{CompTCAM()} to find the optimal TCAM compression. Our algorithm is
based on the ORTC (Optimal Routing Table Constructor) algorithm \cite{Draves99}, that computes
the minimal one dimensional equivalent forwarding table in TCAM.

Intrinsically, our basic idea is to transform the two dimensional table into two conventional one dimensional tables,
one is destination-based and the other is source-based. The action of each prefix in the destination-based (or source-based)
table is the corresponding row (or column) vector in the TD-table. After transforming, the ORTC algorithm can be
applied directly to compress these two tables.

Let $\overrightarrow{TD(f_r(p_d),\cdot)}$ be the row vector related with $p_d$, $\overrightarrow{{TD(\cdot, f_c(p_s))}}$
be the column vector related with $p_s$. Let $\mathcal{DF}(\cdot)$ be a mapping function that maps destination addresses (prefixes)
to row vectors, $\mathcal{SF}(\cdot)$ be a mapping function that maps source addresses (prefixes) to column vectors. Let $ORTC(P,
\mathcal{F}(\cdot))$ be the function by applying ORTC algorithm to the forwarding table, that has prefix set $P$ and action
related with prefix $p\in P$ is $\mathcal{F}(p)$. The input of CompTCAM() is the original forwarding table and
the output of CompTCAM() is the new forwarding table after compression.

\begin{algorithm}[!h]
\footnotesize{
\SetKwInOut{Input}{Initialize}
\SetKwInOut{Output}{Output}
\Output{$\{P_d', P_s', f_r'(\cdot), f_c'(\cdot), TD'(\cdot, \cdot)\}$}
\Begin{
$\forall p_d\in P_d,\mathcal{DF}(p_d)=\overrightarrow{TD(f_r(p_d),\cdot)}$\;
$\{P_d',\mathcal{DF}'(\cdot)\} \gets ORTC(P_d, \mathcal{DF}(\cdot))$\;
$\forall p_d'\in P_d', f'(p_d')\gets f(p_d), \exists p_d\in P_d, \mathcal{DF}'(p_d')=\mathcal{DF}(p_d)$\;

$\forall p_s\in P_s,\mathcal{SF}(p_s)=\overrightarrow{TD(\cdot, f_c(p_s))}$\;
$\{P_s',\mathcal{SF}'(\cdot)\} \gets ORTC(P_s, \mathcal{SF}(\cdot))$\;
$\forall p_s'\in P_s', f'(p_s')\gets f(p_s), \exists p_s\in P_s, \mathcal{SF}'(p_s')=\mathcal{SF}(p_s)$\;
$\forall p_d'\in P_d', p_s'\in P_s', TD'(p_d', p_s')\gets TD(f_r'(p_d'), f_c'(p_s'))$\;
}
\caption{\footnotesize{CompTCAM($P_d, P_s, f_r(\cdot), f_c(\cdot), TD(\cdot, \cdot)$)}}
}
\label{<alg-compression>}
\end{algorithm}

\begin{Theorem}
Algorithm \textup{CompTCAM}() computes the optimal compression, the complexity of \textup{CompTCAM}() is
$w\times N\times M$.
\end{Theorem}
\begin{proof}
For the first part of the theorem, according to ORTC algorithm, $\mathcal{DF}(d)=\mathcal{DF}'(d)$ for any destination
address, $\mathcal{SF}(s)=\mathcal{SF}'(s)$ for any source address, thus the new TwoD-IP forwarding table is equivalent
to the original one. We next prove that CompTCAM() minimizes both destination and source tables by contradiction.
We only give the proof for destination table minimization, proof for source table minimization is similar.

In CompTCAM(), according to ORTC algorithm in \cite{Draves99}, if there exists another compression that produces
$\mathcal{DF}''(\cdot)$ and $P_d''$, and $P_d''$ is smaller than the computed $P_d'$. Then there must be a
destination address $d$ that matches $p_d$, and $\mathcal{DF}''(p_d)\not=\mathcal{DF}'(p_d)$. Thus there must be a
source address $s$, such that $d$, $s$ will match a different index value in the new TD-table.

The complexity of ORTC algorithm is $O(nw)$, where $n$ is the number of rules. In $ORTC(P_d, \mathcal{DF}(\cdot))$ of CompCAM(),
there are $O(N)$ rules, the complexity of the basic comparison operation is $O(M)$. Thus, the complexity of CompTCAM()
is $O(w\times N\times M)$.
\end{proof}

CompTCAM() needs a byte-by-byte comparison between rows and columns. To avoid these wasted comparisons,
we can use \emph{fignerprint}, that is a collision-resistant hash value computed over the rows/columns \cite{Zhu08}.
We use SHA-1 as the collision-resistant hash function, and the collision probability is proved to be much
smaller than hardware error rate \cite{Quinlan02}. Within fingerprints, the complexity of CompTCAM()
can be reduced to be $O(N\times M+(N+M)\times w)$.

Note that CompTCAM() minimizes destination and source tables in TCAM. At the same time, it reduce the size
of TD-table in SRAM, i.e., the row (or column) corresponding to the eliminated destination and source prefixes
will also be eliminated. Next, we try to optimize the storage space in SRAM.

\subsection{Compression in SRAM Space}
Within FIST structure, TD-table and mapping table reside in SRAM storage. Compared to TD-table, mapping
table commonly occupies fixed and small storage space. Thus, we take TD-table as the dominant factor of
SRAM storage.

To be storage efficient, we try to minimize the TD-table. We formulate the problems as following.

\begin{Problem}
\textbf{Optimal TD-table Compression:}
For $\{P_d, P_s, f_r(\cdot), f_c(\cdot), TD(\cdot, \cdot)\}$, find an equivalent forwarding table
such that the storage space in TD-table, i.e., $|\{f_r(p_d)|p_d\in P_d\}|\times|\{f_c(p_s)|p_s\in P_s\}|$ is minimized.
\end{Problem}

\begin{Theorem}
Finding the optimal compressed TD-table is NP-complete.
\end{Theorem}
\begin{proof}
It is obvious that the decision problem of validating a given TD-table
is solvable in polynomial time. Therefore, the optimal TD-table compression
problem is in NP class. To show this problem is NP-hard, we reduce the lossless
data compression problem, which is known to be NP-complete \cite{Storer82}, to it.

The lossless data compression problem is, given a string, find the minimal-length
compressed form of the string. Here we extend the original problem to the two
dimensional case (we call it two dimensional data compression problem), such that
the input string can be a two dimensional string. Two dimensional data compression
problem is also NP-complete, as one dimensional data compression is a special case of it.
Note that two dimensional data compression problem is not equal to the optimal TD-table
compression problem, as rows (columns) in TD-table can be reordered in TD-table.

\begin{figure}[!h]
\includegraphics[width=2in]{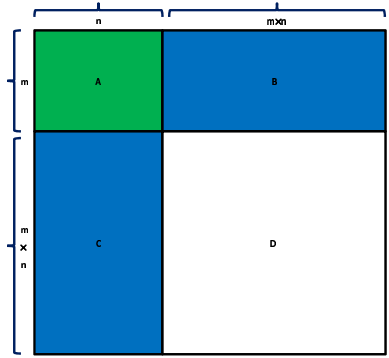}
  \centering
  \caption{\footnotesize{NP-complete proof for the optimal compressed TD-table problem}}\label{fig-optimal-td-np}
\end{figure}

Let $Opt(X)$ be the storage size of the optimal compressed TD-table of TD-table $X$.
Given a $n\times m$ two dimensional string. We construct a $(m+mn)\times(n+mn)$ TD-table,
as shown in Figure \ref{fig-optimal-td-np}. The TD-table is composed of four sub TD-table, $A$, $B$, $C$ and $D$.
Each sub TD-table is independent from the others, i.e., each expressed by a separate symbol
system. $A$ represents the two dimensional string. In $B$, each column $b$ is independent
from other columns, and is the only optimal compressed column, i.e., $Opt(b)>Opt(T\times b)$, where $T$
is a $m\times m$ permutation matrix. In $C$, each row $c$ is independent from other rows, and is
the only optimal compressed row, i.e., $Opt(c)>Opt(T'\times c)$, where $T'$ is a $n\times n$ permutation
matrix. $D$ is an optimal compressed sub TD-table.

Then, we show that by finding an optimal compressed TD-table, we can find a lossless two dimensional
compressed string. This is because if we permutate the sub TD-table $A$, e.g., let the permutation matrix
be $\tilde{T}$, then one of $B$ and $C$ must also be permutated. Without loss of generality, let $B$ be
permutated and the corresponding permutation matrix be $\tilde{T}'$. As $Opt(A)-Opt(\tilde{T}\times A)<m\times n
\le Opt(\tilde{T}'\times B)-Opt(B)$. So permutation on $A$ will never lead to the optimal compressed TD-table.
Thus if we find the optimal compressed TD-table, we can find the lossless compressed data by pick up the
$m\times n$ matrix in the top left corner of the optimal compressed TD-table.

\end{proof}

\subsubsection{Eliminating Duplicated Rows/Columns}

\begin{Observation}
If $\overrightarrow{TD(f_r(p_d), \cdot)}=\overrightarrow{TD(f_r(p_d'), \cdot)}$ (or $\overrightarrow{TD(\cdot,
f_c(p_s))}=\overrightarrow{TD(\cdot, f_c(p_s'))}$), we can merge rows (or columns) of $p_d$ and $p_d'$ (or $p_s$ and
$p_s'$), by setting $f_r(p_d) = f_r(p_d')$ or $f_r(p_d') = f_r(p_d)$ (or $f_c(p_s)=f_c(p_s')$ or $f_c(p_s')=f_c(p_s)$).
\end{Observation}

The observation is true because FIST indirectly points to the index values through row (or column) numbers. If two rows (or
columns) pointed by two destination (or source) prefixes are the same, we can eliminate one of them by making these two
prefixes point to the same row (or column).

Based on this observation, we can eliminate the duplicated rows and columns within FIST structure. The complexity of this process
is $O(log(N)\times N\times M)+O(log(M)\times M\times N) = O(M\times N \times w)$ (because $O(log(N))=O(log(M))=O(w)$).
Within fingerprints, the complexity can be reduced to be $O(M\times N)$. Actually, this process can be combined together
with CompTCAM() to reduce computation time.

\vspace{-0.1cm}
\begin{Theorem}\label{theorem-sram-compress}
	Eliminating the duplicated rows and columns computes the optimal TD-table compression.
\end{Theorem}
\vspace{-0.3cm}

\begin{proof}
	We first prove the equivalence. Without loss of generality,
	we eliminate row first. Let $\{P_d, P_s, f_r'(\cdot), f_c(\cdot), TD'(\cdot, \cdot)\}$
	be the table after eliminating duplicated rows. We have $\overrightarrow{TD(f_r(p_d),\cdot)}=
	\overrightarrow{TD'(f_r'(p_d), \cdot)}$. Let $\{P_d, P_s, f_r'(\cdot), f_c'(\cdot), TD''(\cdot, \cdot)\}$ be the table after eliminating duplicated columns. $TD''$ is a new TD-table. We have $TD'(f_r'(p_d), f_c(p_s)) =
	TD''(f_r'(p_d),f_c'(p_s))$. Thus, $TD''(f_r'(p_d),f_c'(p_s)) = TD(f_r(p_d),f_c(p_s)),
	\forall p_d\in P_d, p_s\in P_s$.
	
	Then we prove the resulted table is the minimum one by contradiction. Assume there exists
	$\{P_d, P_s, \hat{f_r}(\cdot), \hat{f_c}(\cdot), \hat{TD}(\cdot, \cdot)\}$ that is equivalent
	and $|\{\hat{f_r}(p_d)\}|\times|\{\hat{f_c}(p_s)\}|<|\{f_r'(p_d)\}|\times|\{f_c'(p_s)\}|$. Without loss of
	generality, suppose that $|\{\hat{f_r}(p_d)\}|<|\{f_r'(p_d)\}|$. Because $TD''(\cdot, \cdot)$ does not have
	duplicated rows, there must exist $p_d$ and $p_d'$ such that $f_r'(p_d)\not=f_r'(p_d')$ and $\hat{f_r}(p_d)=\hat{f_r}(p_d')$
	(pigeonhole principle). So there must exist $p_s$, such that $TD''(f_r'(p_d),f_c'(p_s))\not=TD''(f_r'(p_d'),f_c'(p_s))$
	and $\hat{TD}(\hat{f_r}(p_d),\hat{f_c}(p_s))=\hat{TD}(\hat{f_r}(p_d'),\hat{f_c}(p_s))$.
	Thus the assumption is wrong. The function of this part is same as rank computation of matrix.
\end{proof}

\subsubsection{Fixed Block Deduplication}

After eliminating the duplicated rows/columns, there still exists duplicated data in TD-table. For example, part of a row
is the same with part of another row. To futher compress the TD-table, we apply fixed block deduplication, which is a
common technique for data deduplication \cite{Meyer12}.

Fixed block deduplication is previously used to eliminate redundancies in data storage (e.g., file system). It breaks file
into chunks that has fixed length, identifies redundant chunk, eliminates all but only one copy, and creates logical pointer
to these chunks so that users can access them as needed \cite{Geer08}.

\begin{figure}[!h]
\includegraphics[width=3in]{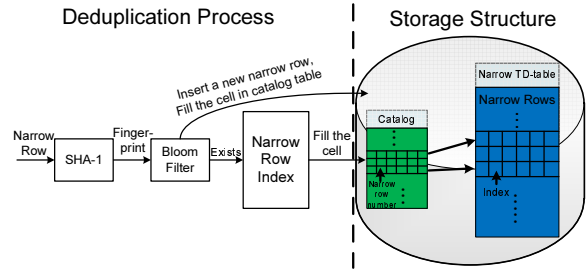}
  \centering
  \caption{\footnotesize{Storage structure and deduplication process for fixed block deduplication}}\label{fig-deduplication-structure}
\end{figure}

\begin{figure}[!h]
\includegraphics[width=3in]{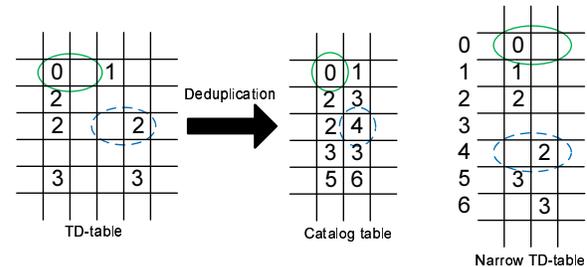}
  \centering
  \caption{\footnotesize{Example of fiexed block deduplication for TD-table}}\label{fig-deduplication-example}
\end{figure}

Our basic idea it to cut the rows in TD-table into fiexed width chunks, called \emph{narrow rows}, i.e., rows that
are shoter than the original rows in TD-table. And we eliminate all duplicated narrow rows, thus only one copy of
each narrow row will be preserved.

As shown in Figure \ref{fig-deduplication-structure}, After deduplicating, we store the narrow rows in a
\emph{narrow TD-table}. Each entry of the narrow TD-table is an indexed next hop, the same with the cell
in TD-table. We also set up \emph{catalog table}, the entry of which points to a row number in narrow TD-table.
Catalog table mapps the narrow rows in TD-table to narrow TD-table. For example, in Figure \ref{fig-deduplication-example},
the TD-table derived from the example in Figure \ref{fig-storage-structure-twod-ip} can be deduplicated into
a narrow TD-table combined with a catalog table. We can see that in Figure \ref{fig-deduplication-example},
the narrow row in the solid cycle can be transformed to be the $0_{th}$ row in the narrow TD-table, and the
narrow row in the dashed cycle can be transformed to be the $4_{th}$ row in the narrow TD-table.

We also show the deduplication process in Figure \ref{fig-deduplication-structure}. We scan the TD-table, and
extract all narrow rows from it. For each narrow row, we first compute the fingerprint of it using SHA-1 function.
With bloom filter \cite{Broder03}, we can judge the narrow row is a duplicated one. If it is not, then we insert
the narrow row into the narrow TD-table, and fill the entry in the catalog table. Else if it is, then we search
in a data structure called \emph{narrow row index}, that organizes all detected <fingerprint, narrow row number>
pairs. Using the search result, we just fill the narrow row number in the corresponding position in the catalog
table.

\ignore{discuss the structure of bloom filter and narrow row index in the implementation section?}

\begin{figure}[!h]
\includegraphics[width=3in]{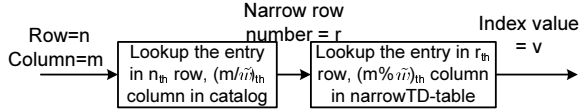}
  \centering
  \caption{\footnotesize{Partial lookup process within the deduplicated storage structure}}\label{fig-deduplication-lookup}
\end{figure}

After using the deduplicated storage structure, the lookup process has to be updated. Let $\tilde{w}$ be the width
(number of cells) of a narrow row. We show the part of the new lookup process in Figure \ref{fig-deduplication-lookup},
which can replace the TD-table lookup step in Figure \ref{fig-lookup-action} and form the whole picture of the new
lookup process. After obtaining the row address and column address, the router can find an entry in the catalog table,
and return a new row address in the narrow TD-table. Using the new row and original column address, the router lookups
the narrow TD-table, and return an index value.
Because of the random access property of the fixed block deplication method, the new lookup process still achieve constant
lookup time. Although it adds one more lookup in SRAM, the influnence on the lookup speed is trivial, especially within
the pipelined lookup model.

\section{FIST Update}\label{sec-update}

Although TD-Saturation() guarantees the correctness of FIST. It needs re-computation
for all conflicted cells, and re-written of them in SRAM when update happens.
Note that although the update is necessary, not all cell need to be cleared
and re-written in this update process. In this section, our objective is
to minimize the number of {\em cell updates}. We use a function $TD(\cdot,\cdot)$ to
denote the TD-table, let $TD(p_d, p_s)$ be the index value of cell $(p_d, p_s)$.

\begin{Problem}
\textbf{Optimal transformation:} Given a TD-table $TD(\cdot,\cdot)$ and an update, find
a new TD-table $TD'(\cdot,\cdot)$, such that $|\{(p_d,p_s)|TD(p_d,p_s)\not=TD'(p_d,p_s)\}|$
is minimized.
\end{Problem}

To achieve this, we will first build a data structure called color tree to organize
the cells. With this color tree, we will develop algorithms for insertion and deletion
where only part of the cells will be updated. We will then prove that our
algorithms indeed minimize the computation cost and the number of cell rewrites.

\subsection{A Color Tree Structure}

Each destination node has a colored tree. The tree is constructed using all source prefixes
in the source table. There are black nodes and white nodes. Intrinsically,
each black nodes represent the cells that are directly {\em set} and white nodes represent
the conflicted cells, i.e., the cells that are not directly {\em set}, but are filled up based
on algorithms, e.g., TD-Saturation().

Let $CT(p_d)$ be the colored tree for $p_d$, let $\mathcal{B}(p_d)=\{p_s|\exists (p_d, p_s, a)\in \mathcal{R}\}$ be the set
of black nodes, let $\mathcal{W}(p_d)=\{p_s|\not\exists (p_d, p_s, a)\in \mathcal{R}\}$ be the set of white nodes. For example,
in Figure \ref{fig-domain-tree}, we show $CT(101*)$, the colored tree for destination prefixes 101*. In it, $\mathcal{B}(101*)=\{****, 11**\}$
and $\mathcal{W}(101*)=\{100*, 101*, 111*\}$.

To compute the optimal transformation upon an update, we first define \emph{domain} of of a black node in colored trees.
Formally,

\begin{Definition}
In a colored tree $CT(p_d)$, the domain of a black node $p_s$ is $\mathcal{D}(p_d, p_s) =\{p_s\}\cup \{p_s'\}$, where $p_s'\in \mathcal{W}(p_d)$
satisfies: 1) $p_s$ is a prefix of $p_s'$; 2) $\not\exists \hat{p_s}\in \mathcal{B}(p_d)$, where $\hat{p_s}$ is a prefix
of $p_s'$ and $p_s$ is a prefix of $\hat{p_s}$.
\end{Definition}

For example, in Figure \ref{fig-domain-tree}, the domain of **** $\mathcal{D}(101*, ****) = \{****, 100*, 101*\}$.
Intuitively, the domain of a black node is the largest sub-tree that roots at itself and does not contain any other
black nodes.

\begin{Theorem}\label{theorem-domain}
When updating rule $(p_d, p_s, a)$, the cell set $\{(p_d, p_s')|p_s'\in \mathcal{D}(p_d, p_s)\}$ is minimum cell set
that should be changed, and all cells in it should be set to be the index value of $a$.
\end{Theorem}
\begin{proof}
We prove the theorem by contradiction. Assume there exists another cell set is smaller than the above cell set,
indicating that the index value of one cell $(p_d,\hat{p_s})$, where $\hat{p_s} \in |\mathcal{D}(p_d, p_s)$, is
not set to be the index value of $a$. Then if a packet matches $p_d$ and $\hat{p_s}$ within FIST, obviously,
it should match the $(p_d, p_s, a)$ rule. Then the cell is set with a wrong index value.
\end{proof}

From Theorem \ref{theorem-domain}, we can see that, through computing the domain, we can compute the optimal transformation
when an update arrives. In the next subsection, we show two update algorithms that compute the optimal transformation.

\begin{figure}[!h]
\begin{minipage}[t]{0.49\linewidth}
\includegraphics[width=1.5in]{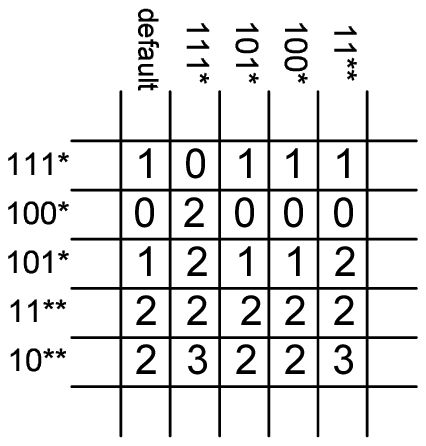}
  \centering
  \caption{\footnotesize{TwoD array after setting all conflicted cells}}\label{fig-twod-array-all-cells}
\end{minipage}
\begin{minipage}[t]{0.49\linewidth}
\includegraphics[width=1.5in]{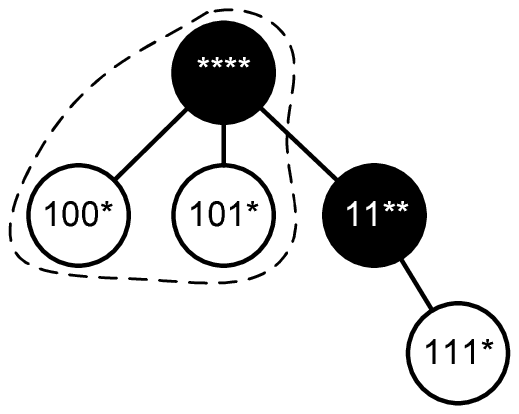}
  \centering
  \caption{\footnotesize{Colored tree $CT(101*)$ for Figure \ref{fig-storage-structure-twod-ip}}}\label{fig-domain-tree}
\end{minipage}
\end{figure}

\subsection{FIST Update Algorithms}\label{sec-update-action}

Here, we define insertion action Insert($p_d, p_s, a$) and deletion action Delete($p_d, p_s$). Update action
Update($p_d$, $p_s$, $a$) can be seen as a deletion followed by an insertion.

Before illustrating these algorithms, we introduce a lemma that simplifies the updating process.

\begin{Lemma}\label{lemma-parent}
If $p_s'$ is the parent of $p_s$ in $CT(p_d)$, and $p_s\in \mathcal{W}(p_d)$, then cells $(p_d, p_s)$ and $(p_d, p_s')$
have the same index value.
\end{Lemma}
\begin{proof}
If $p_s'\in \mathcal{W}(p_d)$, then $p_s$ and $p_s'$ belong to the domain of the same black node. If $p_s'\in \mathcal{B}(p_d)$,
then $p_s$ belong to the domain of $p_s'$. Thus according to Theorem \ref{theorem-domain}, the lemma is proved.
\end{proof}

Algorithm Insert($TD(\cdot,\cdot), p_d, p_s, a$) inserts a rule $(p_d, p_s, a)$ given the TD-table. If $p_d$ (or $p_s$) is
not in the destination (or source) table, router should assign an unused row (or column). When a column is
assigned, we should first find $p_s'$, which is the parent of $p_s$ in the source tree $CT(p_d)$. According to Lemma \ref{lemma-parent},
$(p_d, p_s)$ and $(p_d, p_s')$ have the same index value for all $p_d$. Thus, we copy the column corresponding
to $p_s'$ to the column corresponding to $p_s$. After initializing, through computing the domain, we can find the cells that
should be changed. Finally, $p_d$ (or $p_s$) should be inserted into destination (or source) table if it does not exist.

Algorithm Delete($TD(\cdot,\cdot), p_d, p_s$) delete the rule related with $p_d$ and $p_s$ given the TD-table. At first,
a black node $p_s$ in $CT(p_d)$ is set to be white, thus the nodes in domain $\mathcal{D}(p_s)$ now belongs to a new domain.
For example, in Fig. \ref{fig-delete-example}, after deleting rule $(101*,11**,1.0.0.2)$, node $11**$ is set to be white in
colored tree $CT(101*)$. And nodes $11**$ and $111*$, which belong to the domain of $11**$ before deletion, now belong to
the domain of $****$ in $CT(101*)$. Thus cells $(101*,11**)$ and $(101*,111*)$ should be set to be 1, which is the index value
of $(101*,****)$. After deletion, if there does not exist any rule related with $p_d$ (or $p_s$) any more, we should delete it from
destination (or source) table. And then reclaim the row (or column) resources.

\begin{figure}[!h]
\centering
  \includegraphics[width=3.4in]{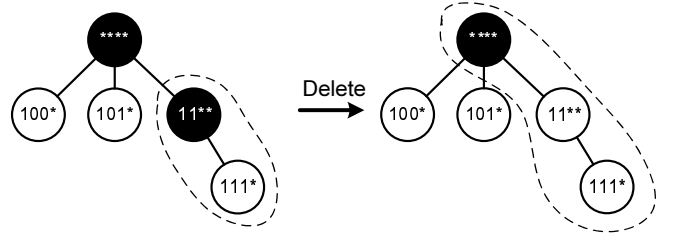}
  \centering
  \caption{\footnotesize{Example of Deletion: }}
  \label{fig-delete-example}
\end{figure}

\begin{algorithm}[!h]
\scriptsize{
\Begin{
\If{$p_d$ does not exist in destination table}
{
Assign a row in TD-table\;
Copy index value of $(p_d, *)$ to cells in the row\;
}
\If{$p_s$ does not exist in source table}
{
Assign a column in TD-table\;
$p_s'\gets$ parent of $p_s$ in $CT(p_d)$\;
copy the column of $p_s'$ to the column of $p_s$\;
}
$TD(p_d, p_s')\gets$ index value of $a$, $\forall p_s'\in \mathcal{D}(p_d, p_s)$\;
\If{$p_s$ ($p_d$) is not in source (destination) table}
{
Insert $p_s$ ($p_d$) into source (destination) table\;}
}
\caption{\footnotesize{Insert($TD(\cdot,\cdot), p_d, p_s, a$)}}
}
\label{<alg-insert>}
\end{algorithm}

\begin{algorithm}[!h]
\scriptsize{
\Begin{
$\tilde{p_s}\gets$ parent of $p_s$ in $CT(p_d)$\;
$TD(p_d, p_s')\gets$ index value of $(p_d, \tilde{p_s})$, $\forall p_s'\in \mathcal{D}(p_d, p_s)$\;
\If{$\not\exists (\hat{p_d},\hat{p_s},\hat{a})\in \mathcal{R}, \hat{p_d}= p_d$}
{
Delete $p_d$ from destination table\;
Reclaim the row of $p_d$\;
}
\If{$\not\exists (\hat{p_d},\hat{p_s},\hat{a})\in \mathcal{R}, \hat{p_s}= p_s$}
{
Delete $p_s$ from source table,
Reclaim the column of $p_s$\;
}
}
\caption{\footnotesize{Delete($TD(\cdot,\cdot), p_d, p_s$)}}
}
\label{<alg-delete>}
\end{algorithm}

\begin{Theorem}
\textup{Insert(}$TD(\cdot,\cdot)$, $p_d$, $p_s$, $a$\textup{)} and \textup{Delete(}$TD$ $(\cdot,\cdot$), $p_d$, $p_s$\textup{)} compute the optimal transformation.
\end{Theorem}
\begin{proof}
The theorem is an immediate result of Theorem \ref{theorem-domain}.
\end{proof}

Thus, the update action causes minimum computation cost, and brings least accesses to TD-table in SRAM.
Beside, with the prevalence of dual-port SRAM, by reading through one port and writing through the
other\footnote{current dual-port SRAM can resolve the read-write collision, i.e., read during write operation
at the same cell \cite{cyclone-handbook}.}, update of TD-table does not have to lock the lookup process.
We can also prove that the update action of FIST is also consistent, i.e., for each rule insertion or deletion,
a packet can only match the rule that would be matched before or after the insertion or deletion \cite{Wang04}.
Due to page limit, we omit it here.

\section{Practical Considerations}\label{sec-practical-consideration}
\subsection{Reducing Update Burden on TD-table}
\label{sec-improvements-update}
Although TD-table update will not influence the lookup process, a single rule insertion/deletion may still cause $O(M)$
write operations at SRAM. An update process, in the worst case, will cause updating on all cells in a row of TD-table.
For example, if we update $(11**,****)$ in Figure \ref{fig-storage-structure-twod-ip} with index value 1, then all cells
in the $3_{rd}$ row should be updated with index value 1.

When $M$ is very large, it may exceed the ability of SRAM to handle these updates. For example, if
there are 10,000 source prefixes, the network produces over 500 updates on the default next hops of different destination
prefixes per second. In the worst case, there will be over 5 million/second write operations into TD-table, which exceed
the maximum clock rate of SRAM.

The main reason for the large number of update operations on TD-table is that, the default next hop of all destination
prefixes is stored as full wildcard in source table. First, the full wildcard resides at the root node of the source tree,
once updated, it will cause a lot of updates subsequently. Second, the default next hop changes frequently, because it
has to change when connectivity information of the corresponding destination prefix changes.

Thus, we propose to isolate default next-hop from source table, i.e., it is not stored in source table. Rather than being
matched when the full wildcard is hit in source table, the default next-hop is matched when none entry in source table is
matched. In Section \ref{sec-implementation}, we will illustrate this in detail.

After removing the full wildcard from the source table, we believe the update frequency of TD-table will be low, with the following
two facts: 1) the update of non-connectivity rule will be slow, i.e., it does not have to respond instantly to the changes
of network topology; 2) most prefixes in the current forwarding tables are near leaf nodes in prefix trees \cite{Basu05}, indicating
that we only need to update a few cells during most rule updates.

After removing the full wildcard, the source tree may be divided into {\em source forest}, which has a similar definition
with source tree except its forest structure. For example, in Fig \ref{fig-source-forest}, we show the source forest after removing
the full wildcard. We can also define {\em colored forest} (denoted by $CF(p_d)$) and re-define domain in a similar way.

However, unlike in the colored tree,
where each node has a black ancestor at least, because the root node is black. In colored forest, a node could have none black
ancestor without black root node. Thus, in colored forest, a white node may do not belong to the domain of any black node.
For example, in Fig \ref{fig-source-forest}, the shaded white node 100* and 101* do not belong to the domain of any black node.
For the white node $p_s$ that do not belong to the domain of any black node in $CF(p_d)$, cell $(p_d,p_s)$ is {\em invalid},
i.e., the cell does not have any index value and should not be matched. Fig. \ref{fig-twod-array-after-isolating} shows
the TD-table after removing the full wildcard from source table.

After that, we can revise update action, including insertion and deletion, by replacing ``tree" with ``forest".

\begin{figure}[!h]
\begin{minipage}[t]{0.49\linewidth}
\centering
  \includegraphics[width=1.5in]{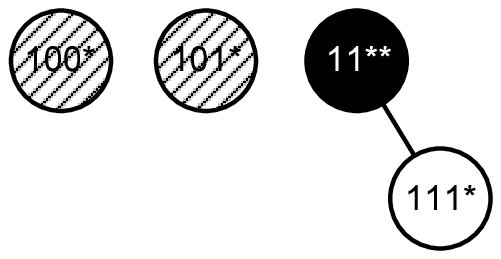}
  \centering
  \caption{\footnotesize{Source forest after removing the full wildcard}}
  \label{fig-source-forest}
\end{minipage}
\begin{minipage}[t]{0.49\linewidth}
\includegraphics[width=1.5in]{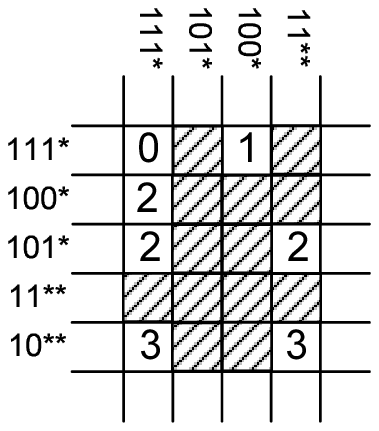}
  \centering
  \caption{\footnotesize{TwoD array after removing the full wildcard in source table}}
  \label{fig-twod-array-after-isolating}
\end{minipage}
\end{figure}

\section{Implementation}\label{sec-implementation}
As a proof-of-concept, we implement the FIST forwarding table structure on a commercial router, Bit-Engine 12004,
which supports four linecards. In each linecard, there are a CPU board (BitWay CPU8240 that works at 100MHz), two
TCAM chips (IDT 74K62100, can accommodate 512K IPv4 entries at most), an FPGA chip (Altera EP1S25-780), and several
cascaded SRAM chips (IDT 71T75602) associated with the TCAM chips. Inside the FPGA chip, there exists internal SRAM
memory.

Our implementation is based on existed hardware, and does not need any new hardware. We re-design the hardware through
rewriting about 1500 lines of VHDL code (not including C code) of the original destination-based version.

\subsection{Router Framework}
In Fig. \ref{fig-router-framework}, we show the framework of our router design. The major changes are in data plane.
In data plane, the FPGA receives the packets from the interface module, extracts the packet head, and request the TCAM
module. Due to resource limit, we place both destination and source table in different blocks of one TCAM, and the FPGA
requests the TCAM module twice (the first in destination table, and the second in source table) to access these two tables.
Although this will increase the delay per lookup in our implementation, many processors (e.g., NetLogic NL10K) now support two lookups
in parallel, thus this will not become the bottleneck of lookup process in the future.

The TCAM module will output the matched prefix, and through the TCAM associated SRAM, FPGA will get the matched result, e.g., row or
column address of matched prefix. Then the FPGA will compute the address of the cell in TD-table, which resides in a block of an
internal SRAM of the FPGA. After getting the index, FPGA accesses the mapping table, which resides in another block of the SRAM of
FPGA. Then FPGA gets the next hop information, and delivers the packets to the next processing module, switch co-process module,
which will switch the packet to the right interface.

We also design the control interfaces for control plane to access and update the forwarding table. In the control plane, we
store destination prefixes in a patricia trie \cite{Knuth98}, source prefixes in another patricia trie. We store the row and column
addresses in the nodes of each patricia trie, and also store each rule in a two dimensional array.

\begin{figure}[!h]
\centering
  \includegraphics[width=2.5in]{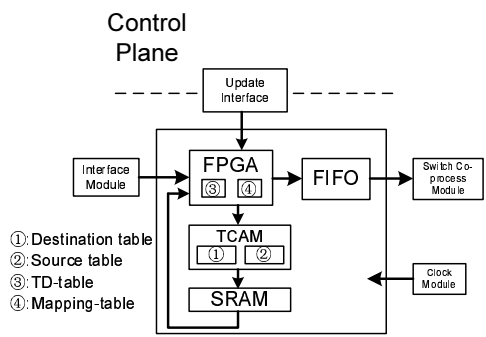}
  \centering
  \caption{\footnotesize{The framework of router design}}
  \label{fig-router-framework}
\end{figure}

\begin{figure}[!h]
\centering
  \includegraphics[width=3in]{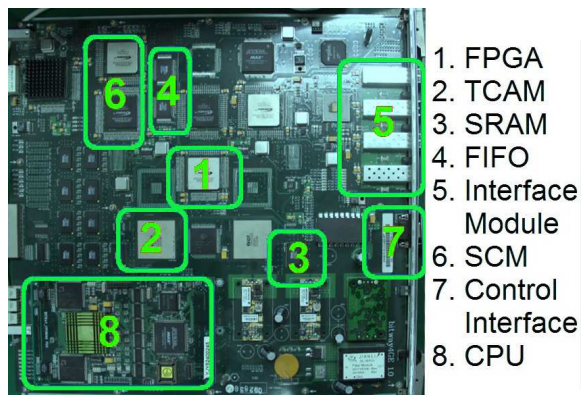}
  \centering
  \caption{\footnotesize{Implementation of FIST on the linecard of Bit-Engine 12004}}
  \label{fig-router-implementation}
\end{figure}

\subsection{A Scalable FIST Design}

We implemented the FIST structure, as shown in Figure \ref{fig-storage-structure-twod-ip}, on the linecard. Besides, for
better scalability, we incorporated the improvements mentioned in Section \ref{sec-improvements-storage} and
\ref{sec-improvements-update}, such that FIST can accommodate more destination/source prefixes and allow more frequent updates.
Within the improvements, the format of the SRAM units pointed by source table remains the same, i.e., storing only the column
address. However, the format of the SRAM units pointed by destination table changes: 1) it has an indicator bit, which is set only
if there is a row in TD-table for the corresponding destination prefix, such that we can reduce the SRAM space of TD-table (see
Section \ref{sec-improvements-storage}); 2) it stores the index value of the default next hop for the corresponding destination
prefix, such that the burden on TD-table caused by updates can be reduced (see Section \ref{sec-improvements-update}).

\begin{figure}[!ht]%
\centering
\subfloat[Format of SRAM units pointed by source table]{
\centering
\includegraphics[width=1.3in]{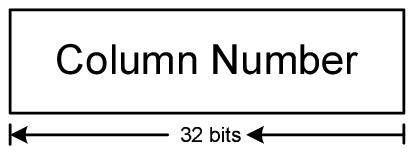}
\label{fig-source-sram-format}%
}
\subfloat[Format of SRAM units pointed by destination table]{
\centering
\includegraphics[width=1.5in]{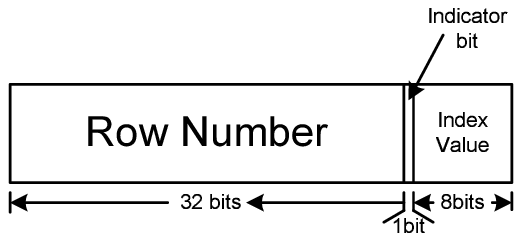}
\label{fig-destination-sram-format}%
}
\caption{\footnotesize{The format of the SRAM units pointed by TCAM entries}}
\end{figure}

Within the new structure, the lookup process also changes. After TCAM matching in destination and source tables and obtaining
the SRAM unit corresponding to the matched prefixes. Router checks the indicator bit, if the indicator bit is unset, then router
gets the index value of the default next hop directly. Else if none source prefix gets matched, then router gets the index value
of the default next hop. Else if a source prefix is matched, then router accesses the cell $(p_d, p_s)$ (assume $p_d, p_s$
are the matched destination and source prefixes) in TD-table. If the cell is invalid, then router gets the index value of
the default next hop, else the router gets the index value of cell $(p_d, p_s)$. Using the obtained index value, router looks
up in mapping table, and gets the next hop information. We show the new lookup process in Figure \ref{fig-new-lookup-process},
note that compared to the original lookup process in Figure \ref{fig-lookup-action}, all new steps are processed in CPU, indicating
that there is none additional accesses to TCAM or SRAM.

\begin{figure*}[!ht]%
\centering
\centering
\includegraphics[width=5in]{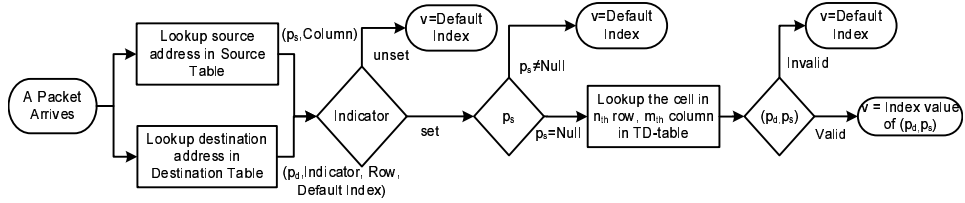}
\caption{\footnotesize{Lookup process of our implementation}}
\label{fig-new-lookup-process}%
\end{figure*}

\subsection{Fixed Block Deduplication}

We use bloom filter to accelerate the deduplication process. Within bloom filter, there exists a summary vector \cite{Zhu08},
which is a vector of $k$ bits. 

\section{Evaluation}\label{sec-evaluate}
\subsection{Evaluation Setup}

\begin{figure}[!h]
  \includegraphics[width=3.3in]{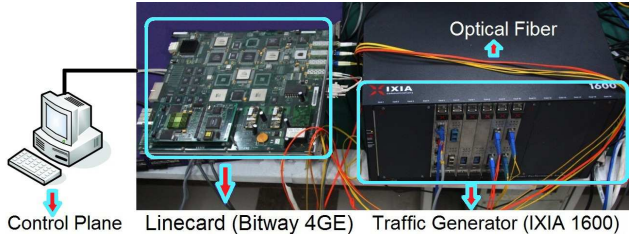}
  \centering
  \caption{\footnotesize{Evaluation environment}}\label{fig-connection-diagram}
\vspace{-0.2cm}
\end{figure}

In Figure \ref{fig-connection-diagram}, we show the connection diagram of the evaluation environment. There are three components,
a PC host (in this paper, the CPU of the PC host is Intel Core2 Duo T6570) that acts as the control plane of a router, a 4GE linecard
that has been equipped with both ACL-like and FIST forwarding table structure and a traffic generator (IXIA 1600). Using
optical fibers, the linecard is connected with the traffic generator, and using a serial cable, the linecard is connected with the
PC host. The traffic generator sends packets of minimum 64 bytes (including 18 bytes of Ethernet Header) at full speeds, i.e., 4Gbps.
The linecard receives the packets, lookups them in the forwarding table, and sends them back to the traffic generator. The traffic
generator can summarize the sending and receiving rate.

We control the forwarding table by the PC host through the serial cable. We update the forwarding table using Algorithm
$Insert(p_d, p_s, a)$ and $Delete(p_d, p_s)$ through the pre-defined interfaces to access hardware on the PC host. We test update
at different frequency, i.e., 100, 1,000, and 10,000 updates per second. The TCAM memory is constructed according to the L-algorithm
\cite{Shah01}, i.e., prefixes of the same length are clustered together and there exists free space between different clusters, to
guarantee fast updates in TCAM. We pre-allocate 1000 positions for each prefix cluster of different length initially.

\subsection{Data Sets}
To evaluate our FIST structure, we consider two scenarios, and generate forwarding table data sets, update sequence data sets
within these scenarios.

\begin{figure}[!h]
  \includegraphics[width=3.3in]{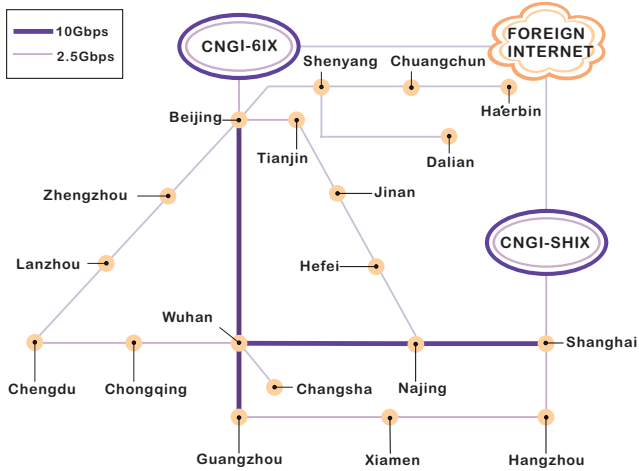}
  \centering
  \caption{\footnotesize{CERNET2 topology}}\label{cernet2_topology}
\vspace{-0.2cm}
\end{figure}

\subsubsection{Policy Routing in CERNET2}
CERNET2 has two international exchange centers connecting to the Internet, Beijing (CNGI-6IX) and Shanghai (CNGI-SHIX).
However, during operation, we found that CNGI-6IX is very congested with an average throughput of 1.18Gbps in February
2011; and CNGI-SHIX is much more spared with a maximal throughput of 8.3Mbps at the same time. We want to move the out-going
International traffic of three universities, i.e., THU (in Beijing, with 38 prefixes), HUST (in Wuhan, with 18 prefixes) and
SCUT (in Guangzhou, with 28 prefixes) to CNGI-SHIX (Shanghai portal).

In this scenario, we collects the prefix and FIB information from CERNET2. There are 6973 prefixes in the FIB of CERNET2, and
among them there are 6406 foreign prefixes. We construct three policy forwarding tables on three routers, i.e., Beijing, Wuhan
and Guangzhou (we call each forwarding table PR-BJ, PR-WH, PR-GZ).

To obtain the update sequence, we set the initial two dimensional rule set to be empty, and add all rules into the forwarding
table at some time point. We generate the update sequence on the router of Wuhan in this way, to simulate a common scenario,
where ISPs decide to carry out a policy at some time point. We show the number of rules in each forwarding table, and number
of updates in each update sequence in Table \subref*{tab-rule-number}.

\begin{table}[h!]
\scriptsize{
\subfloat[Number of rules in each forwarding table]{
  \centering
  \begin{tabular}{c|c}
    \toprule
    Forwarding table  &  \# of Rules \\
    \midrule
    PR-BJ       &  250366    \\
    PR-GZ       &  186306   \\
    PR-WH       &  365674  \\
    LB-MO       &  7118   \\
    LB-AF       &  7342   \\
    LB-NI       &  7410  \\
    \bottomrule
  \end{tabular}
  \label{tab-rule-number}
}
\subfloat[Number of updates in each update sequence]{
  \centering
  \begin{tabular}{c|c}
    \toprule
    Update sequence  &  \# of Updates\\
    \midrule
    PR          & 365674 \\
    LB          & 475773\\
    \bottomrule
  \end{tabular}
  \label{tab-update-number}
}
}
\caption{\footnotesize{Data sets overview}}
\label{tab-datasets-overview}
\end{table}

\subsubsection{Load Balancing in CERNET2}
To further balance the load in between CNGI-6IX and CNGI-SHIX, we need a more dynamic load balancing mechanism in the future.
We collect about one Tera-Bytes of traffic data during one month (Jan, 2012) from three routers (i.e., Beijing, Shanghai and Wuhan)
by NetFlow. In Figure \ref{fig_cernet2_utilization} (the Y-axis has been anonymized), we also show the bandwidth utilization of
both CNGI-6IX and CNGI-SHIX during the month. We can see that CNGI-6IX is much more congested than CNGI-SHIX.

\begin{figure}[!h]
  \includegraphics[width=3in]{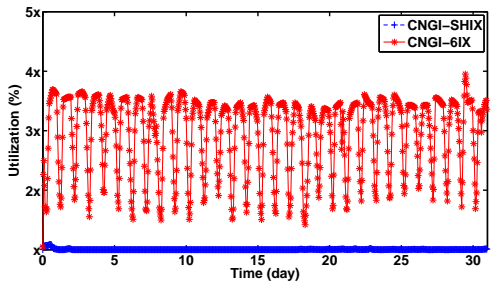}
  \centering
  \caption{\footnotesize{Bandwidth utilization of CNGI-6IX and CNGI-SHIX in CERNET2}}\label{fig_cernet2_utilization}
\vspace{-0.2cm}
\end{figure}

We first process the traffic data, such that each out-going international micro-flow, that is identified by their source and
destination addresses, is aggregated into a macro-flow, that is identified by an source and an destination prefix (here, we
use the LMF rule for aggregation). Then we try to redistribute each macro flow to different exchange centers, such
that load is optimally balanced. The problem can be reduced to Multi-Processor Scheduling problem \cite{Blazewicz86} and is NP-hard.
To solve the problem, we use the greedy first-fit algorithm, which assigns each macro flow to the exchange center with the least
utilization, and achieves an approximation factor of 2.

We construct three load balancing forwarding tables, each at a different time points, i.e., 6:00 morning, 2:00 afternoon and 10:00
evening during Jan 15, 2012 on the router of Wuhan (we call each forwarding table LB-MO, LB-AF, LB-EV). We also show the number of
each forwarding table in Table \ref*{tab-update-number}, in which we can see that LB-EV is the largest one, because more traffic should
be move to CNGI-SHIX when 10:00 at night is the peak traffic point during one day. We also generate the update sequence by computing
a new load balancing scheme every hour.

\subsection{Evaluation Results}

\subsubsection{Forwarding Table Size}
We evaluate the storage space that FIST consumes for all forwarding tables, and the storage space after compression and
after adopting non-homogeneous structure. As a comparison, we also set the ACL-like structure as a benchmark.
In Figure \ref{fig-forwarding-table-size}, we show the size of each forwarding table, which can be separated into TCAM
and SRAM storage, within different storage structures.

\textbf{Trivial FIST and ACL-like Structure:} In Figure \subref*{fig-tcam}, we can see that for all forwarding tables, FIST consumes
only half of the TCAM space that ACL-like structure consumes. For the data sets within the policy routing scenario,
FIST costs much less TCAM storage, e.g., within PR-WH, FIST consumes about 1Mb, while ACL-like structure consumes more
than 72Mb TCAM storage space. This is because in our policy routing scenario, the forwarding table is very \emph{dense}, i.e.,
many rule share the same destination or source prefix. FIST store only once for each destination or source prefix, while
ACL-like structure may store multiple times for the same destination (or source) prefix if it is associated with multiple
source (or destination) prefixes.

In Figure \subref*{fig-sram}, we can see that, within PR-BJ, PR-GZ and PR-WH, FIST consumes less SRAM space than
ACL-like structure. However, within LB-MO, LB-AF and LB-EV, FIST consumes more SRAM space ACL-like
structure. This is because in the policy routing scenario, the forwarding table is much denser, the rules are congregated in a
few source and destination prefixes, i.e., the prefixes of THU, HUST and SCUT. However, in the load balancing scenario, rules
span across many destination and source prefixes, thus the FIST structure consumes much more SRAM space.

\textbf{Compression:}
In Figure \subref*{fig-compression-tcam} and \subref*{fig-compression-sram}, we show the consumed TCAM and SRAM storage
space within FIST after compression, first by Compress-DS() and then by Compress-TD(). We also compress the forwarding
tables within ACL-like structure, i.e., minimize the number of rules. Note that within ACL-like
structure, we can not further reduce SRAM storage space after minimizing the TCAM storage space.
In Figure \subref*{fig-compression-tcam},
we can see that about 20\%-30\% TCAM storage space can be saved through compression. After Compress-DS(), the TCAM storage
space consumed by FIST is still much smaller than the TCAM storage space consumed by ACL-like structure.
Compress-TD() does not effect on TCAM storage, as it only modifies the row (or column) number of destination (or source)
prefixes.
In Figure \subref*{fig-compression-sram}, we show the consumed SRAM storage space after compression. We can see that the
percentage of SRAM that can be saved by Compress-SD() and compressing ACL-like forwarding table is similar with
the percentage of TCAM that can be saved, i.e., 20\%-30\%. However, Compress-TD() has a considerable effect on the SRAM
storage space within FIST, this is because there are high redundancies in the TD-table. For example, on PR-WH, we carry
out the same policy on all source prefixes in source table, thus their corresponding columns in TD-table can be merged.

\textbf{Non-Homogenous Structure:}
In Figure \subref*{fig-non-homo-tcam}, we show the consumed TCAM storage space within non-homogeneous structure.
Non-homogeneous FIST structure does not save TCAM storage space, because non-homogeneous structure only separates
the destination table into two parts. Non-homogeneous ACL-like structure does save TCAM storage space,
especially for the load balancing scenario. This is because the width of an TCAM entry can be reduced after store
destination only rules separately. However, because the width of a TCAM entry is fixed, we can only physically
(instead of logically) divide the table into two parts within ACL-like structure. In contrast, within
FIST, we can flexibility logically divide the table into two parts.

In Figure \subref*{fig-non-homo-sram}, we show the consumed SRAM storage space within non-homogeneous structure.
Within non-homogeneous FIST structure, SRAM storage space can be saved. Within PR-BJ, PR-GZ and PR-WH, about
7\% SRAM space can be saved after adopting non-homogeneous structure, because about 7\% destination prefixes are
not foreign prefixes, and does not have to be moved. However, for LB-MO, LB-AF and LB-EV, the SRAM space can be reduce
to be 3\% of the SRAM space consumed by homogeneous structure. This is because in the load balancing scenario, only
traffic of a small number of destination prefixes have to be diverted to another path. For example, within LB-EV,
only traffic towards 59 destination prefixes has to be diverted. Within non-homogeneous ACL-like structure,
SRAM storage space is not saved. After adopting non-homogeneous structure, FIST costs less SRAM storage than ACL-like structure.

\textbf{Combine Non-Homogenous Structure with Compression:}
In Figure \subref*{fig-final-tcam} and \subref*{fig-final-sram}, we apply both non-homogenous structure and compression
techniques to all forwarding tables. The resulting tables get smaller than all previous tables. Here, we focus on SRAM
storage space because non-homogenous structure has no effect on TCAM storage space.  We can see that the improvement is
small compared to compression only, this is because the TD-table is already very small, and negligible as compared to
other consumed SRAM storage space. However, the improvement is quite large compared to non-homogenous structure only,
because there still exist high redundancies after adopting non-homogenous structure.

\begin{figure*}[!ht]%
\centering
\subfloat[TCAM storage space for FIST and ACL-like sturcture]{
\centering
\includegraphics[width=1.8in]{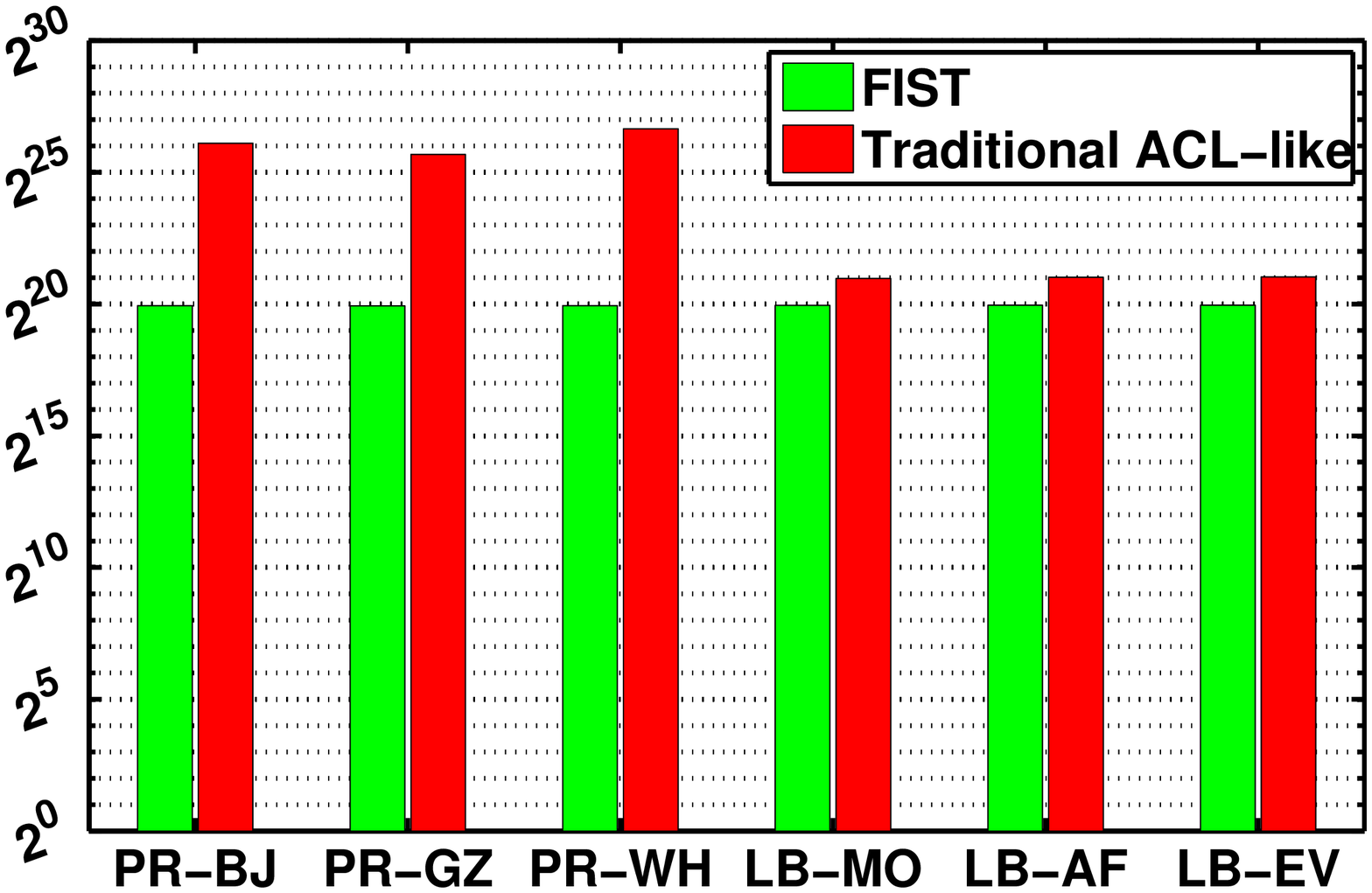}
\label{fig-tcam}%
}
\subfloat[TCAM storage space after compression]{
\centering
\includegraphics[width=1.8in]{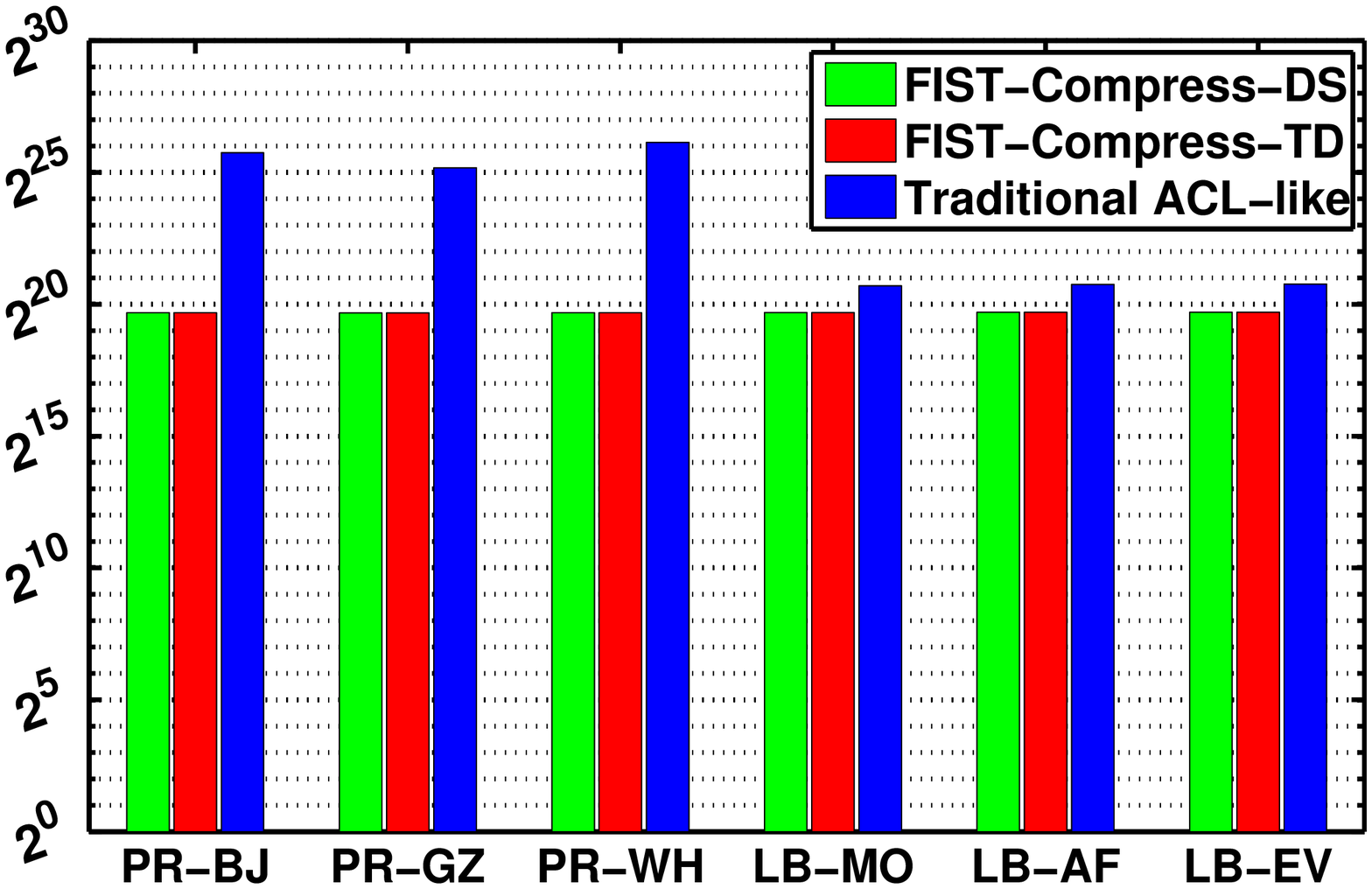}
\label{fig-compression-tcam}
}
\subfloat[TCAM storage space with non-homogeneous structure]{
\centering
\includegraphics[width=1.8in]{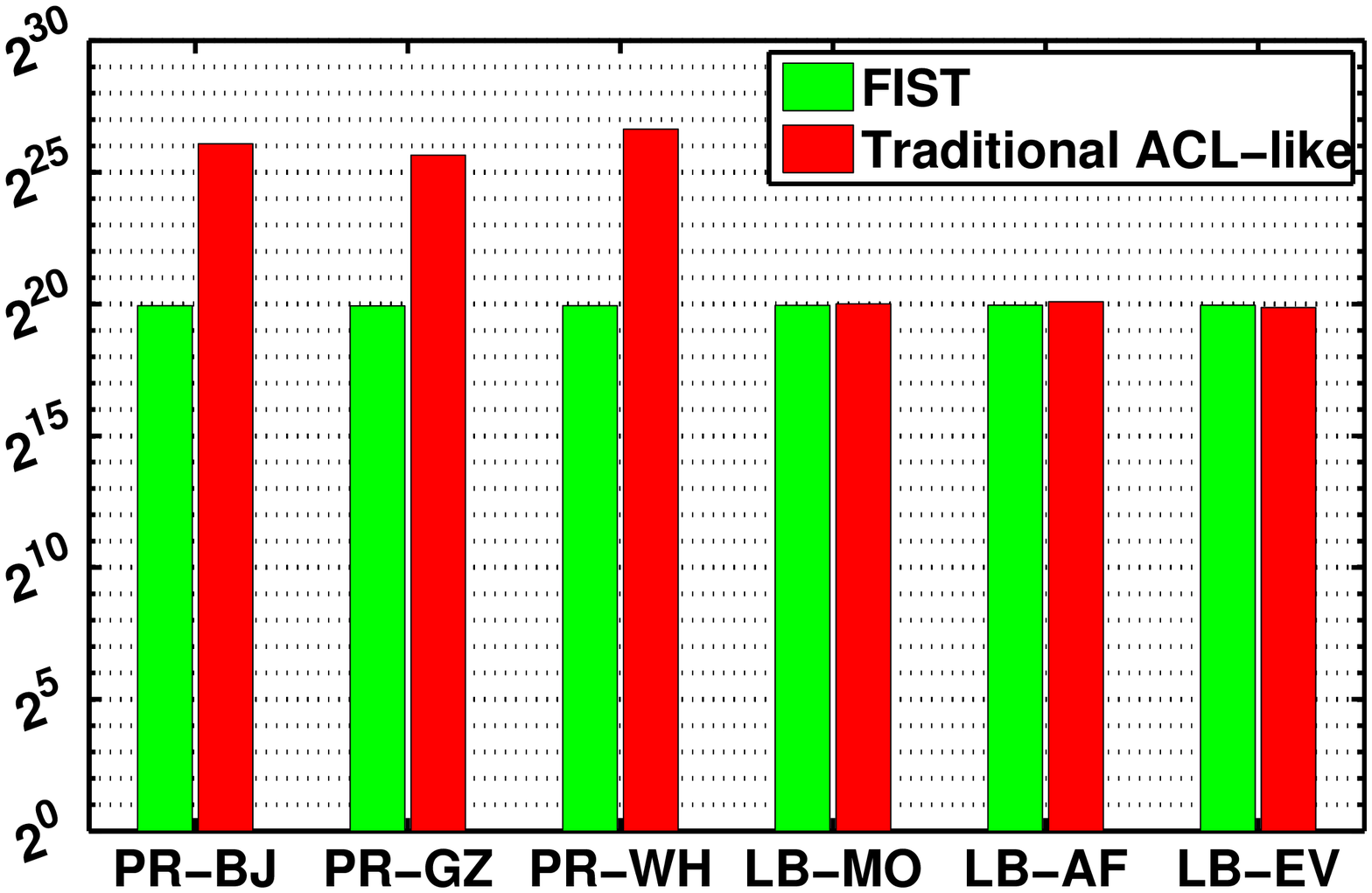}
\label{fig-non-homo-tcam}
}
\subfloat[TCAM storage space after compression, with non-homogeneous structure]{
\centering
\includegraphics[width=1.8in]{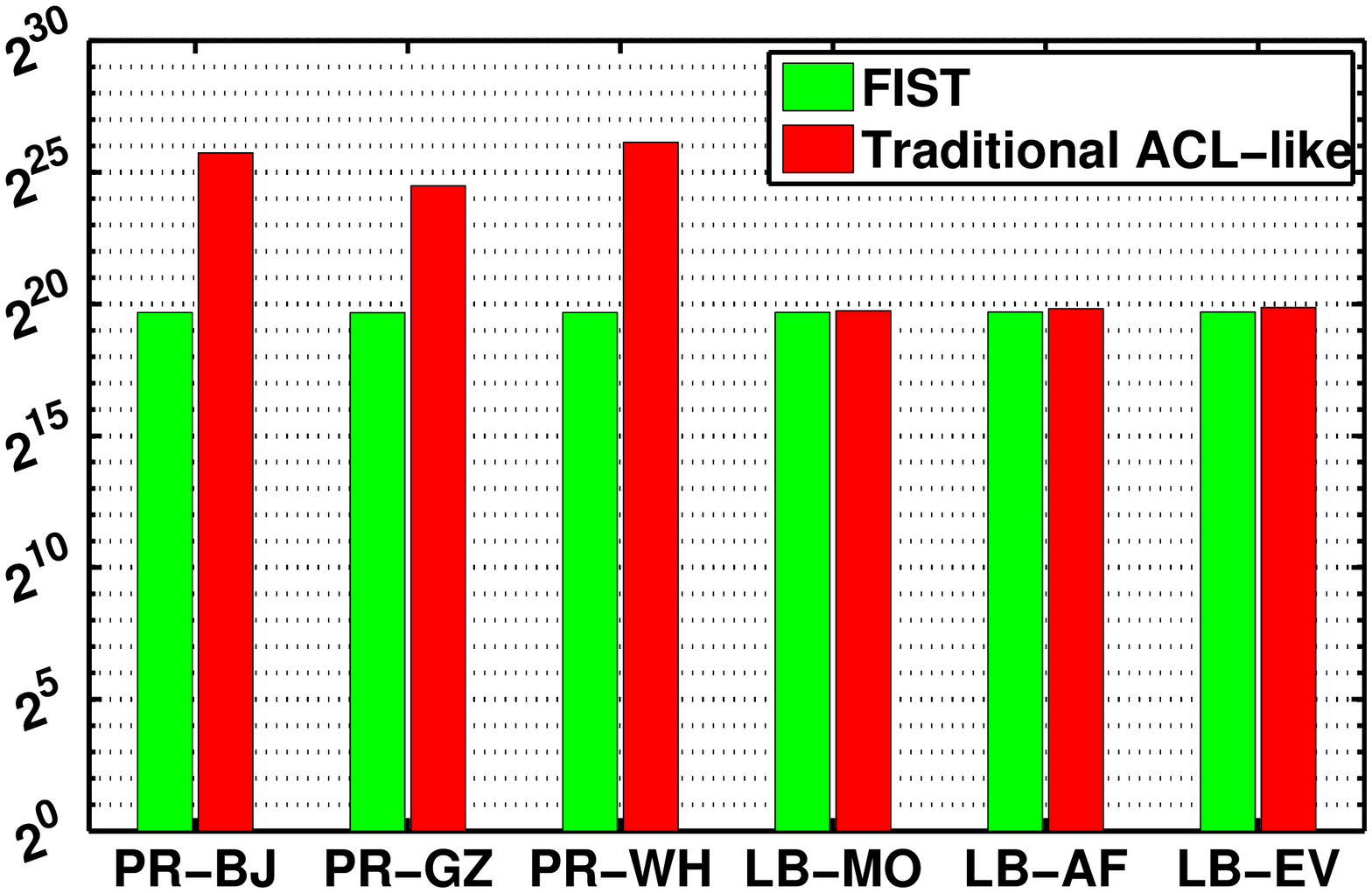}
\label{fig-final-tcam}
}
\\
\subfloat[SRAM storage space for FIST and ACL-like sturcture]{
\centering
\includegraphics[width=1.8in]{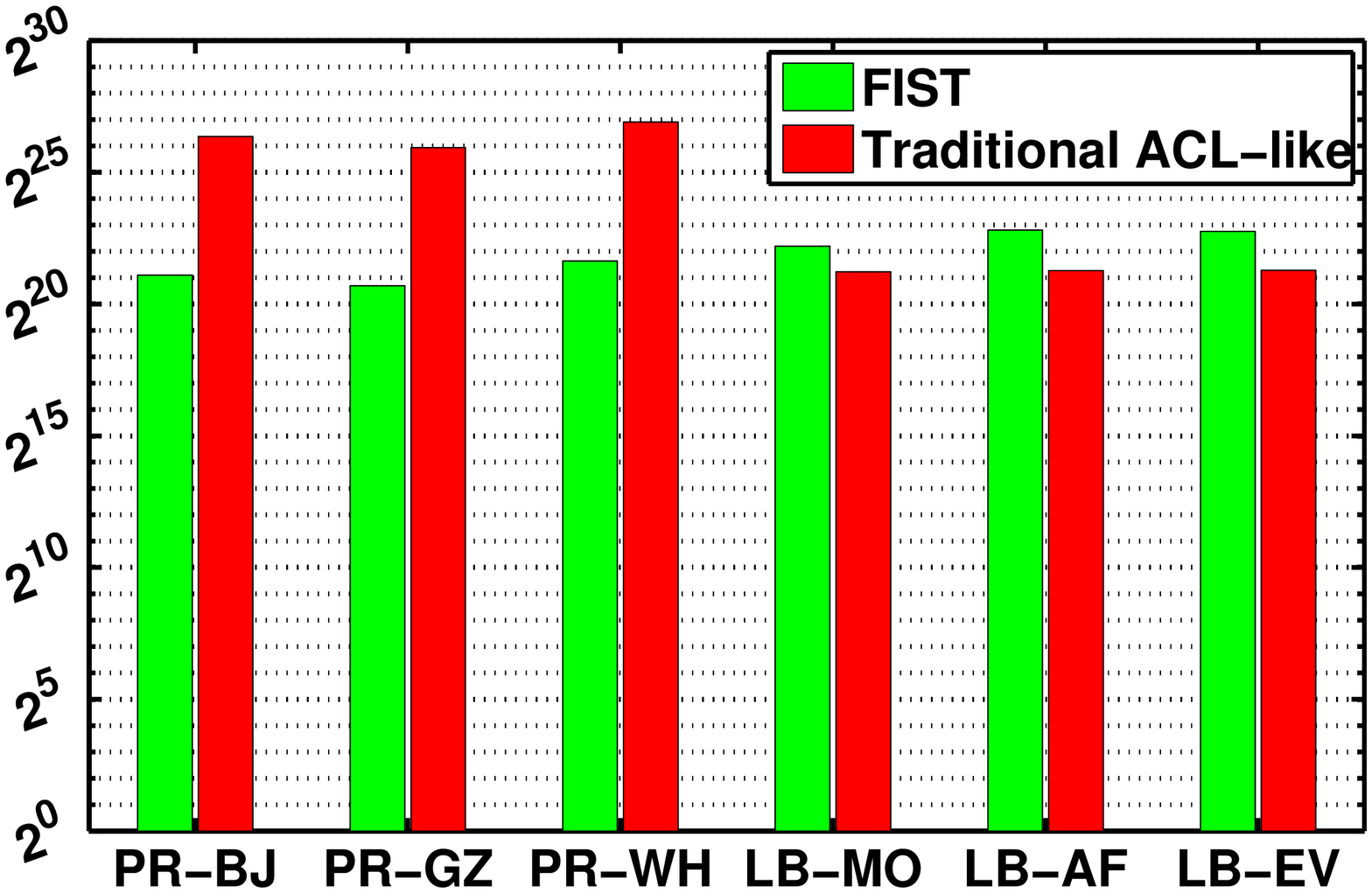}
\label{fig-sram}%
}
\subfloat[SRAM storage space after compression]{
\centering
\includegraphics[width=1.8in]{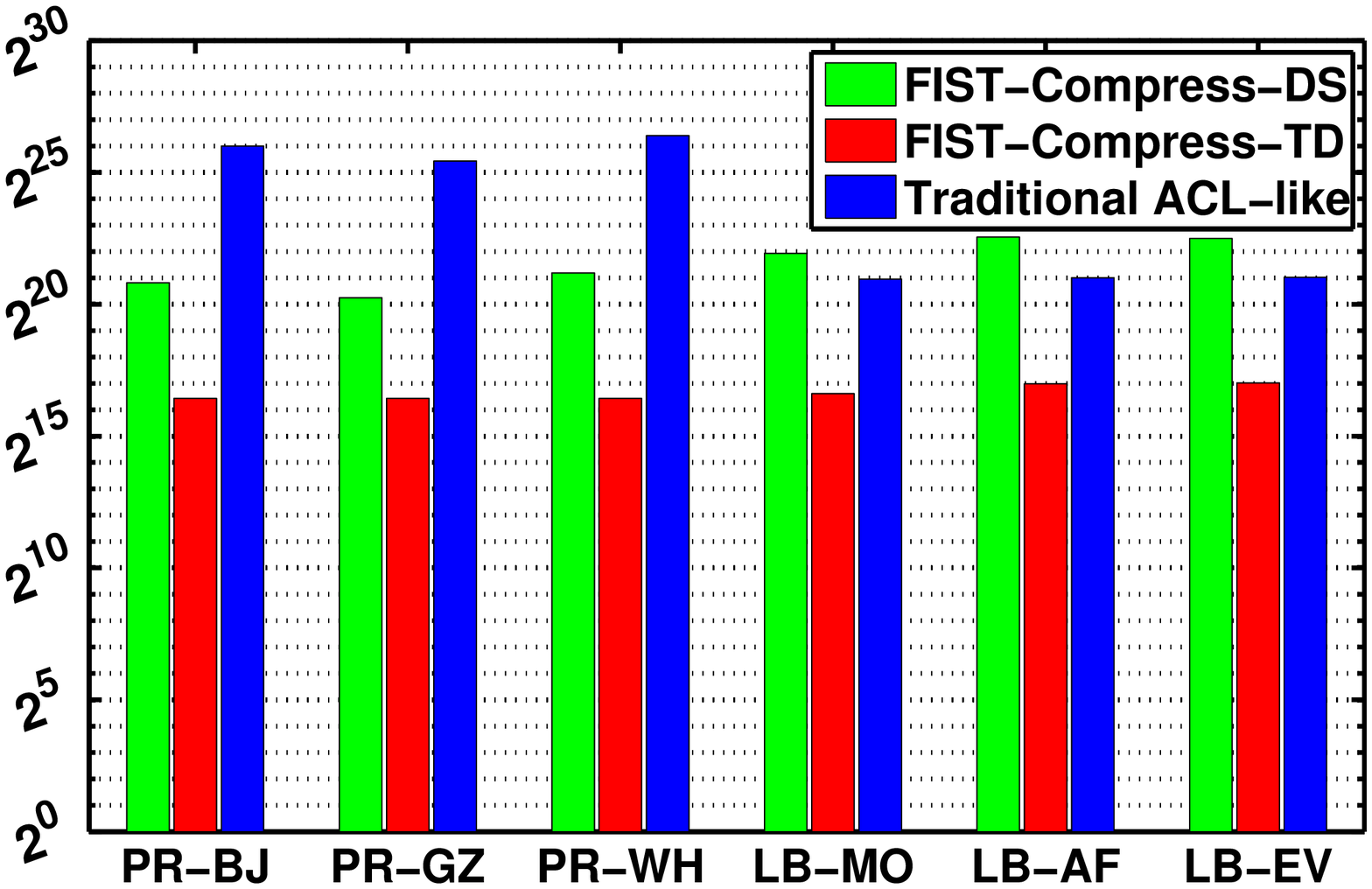}
\label{fig-compression-sram}
}
\subfloat[SRAM storage space with non-homogeneous structure]{
\centering
\includegraphics[width=1.8in]{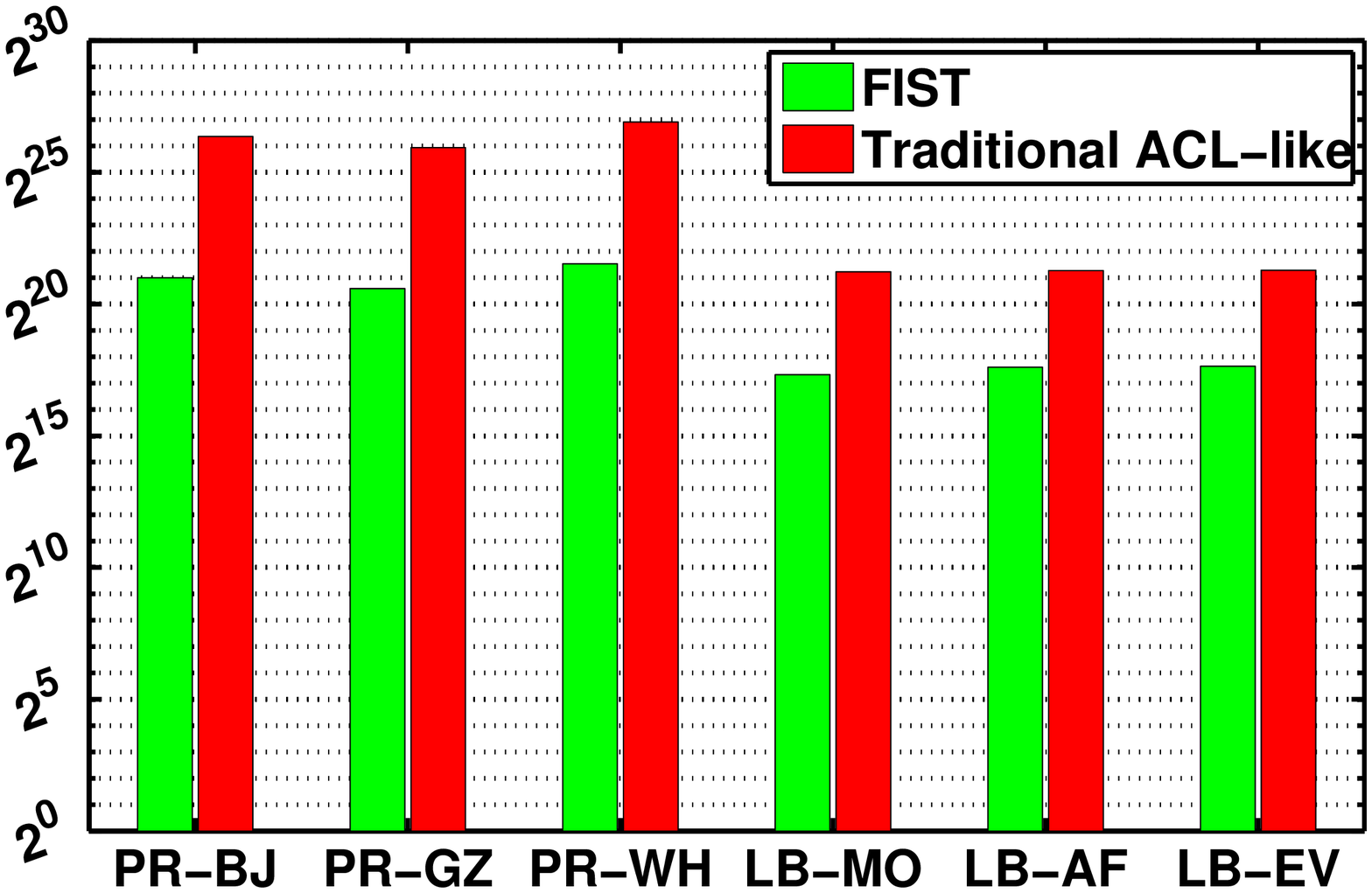}
\label{fig-non-homo-sram}
}
\subfloat[SRAM storage space after compression, with non-homogeneous structure]{
\centering
\includegraphics[width=1.8in]{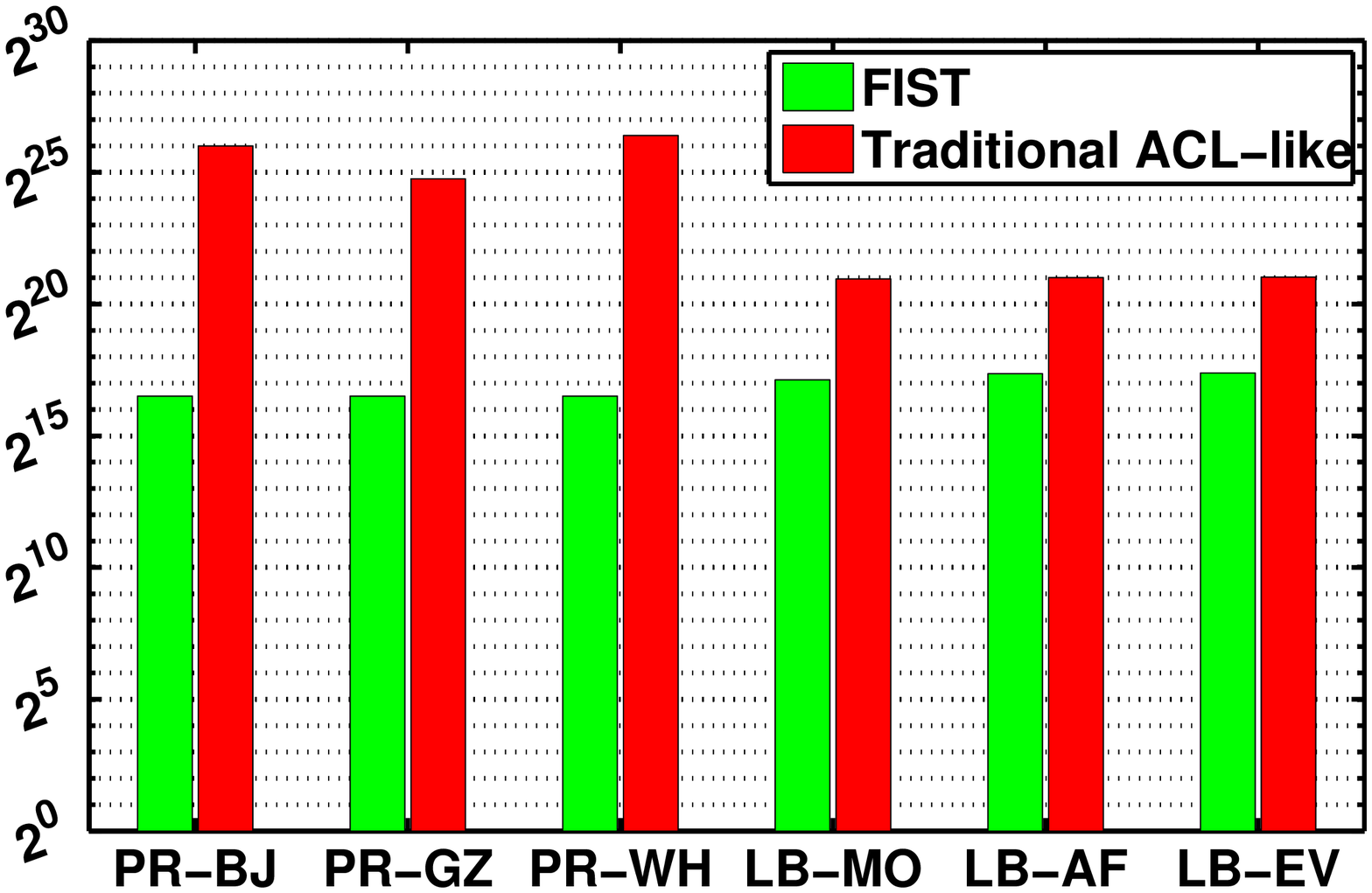}
\label{fig-final-sram}
}
\caption{\footnotesize{Size of each forwarding table}}
\label{fig-forwarding-table-size}
\vspace{-0.3cm}
\end{figure*}

\subsubsection{Lookup Speed and Update}
\textbf{Lookup Speed:}
In Figure \ref{fig-evaluation-none-update}, we show the lookup speed without update. We can see that without
update, both sending and receiving rates reach line speeds (note that Ethernet frame contains 8 bytes of preamble
and 12 bytes of gap, thus the maximum sending rate is $4\times\frac{64}{64+20} \approx 3.0476$Gbps). We also
look into the data traces, and find there does not exist packet loss.

\begin{figure}[!ht]%
\centering
\includegraphics[width=3.1in]{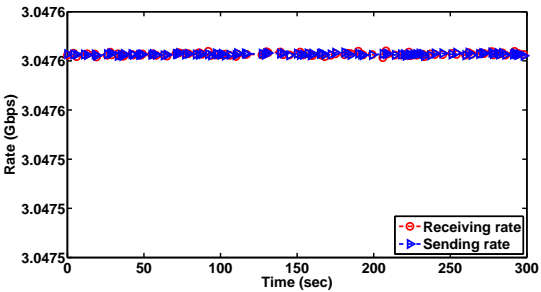}
\caption{\footnotesize{Lookup speed without update}}
\label{fig-evaluation-none-update}%
\end{figure}

\textbf{Number of Accesses to TCAM During Update:}

\begin{figure*}[!ht]%
\centering
\subfloat[Number of accesses to TCAM per 100 updates]{
\centering
\includegraphics[width=2.1in]{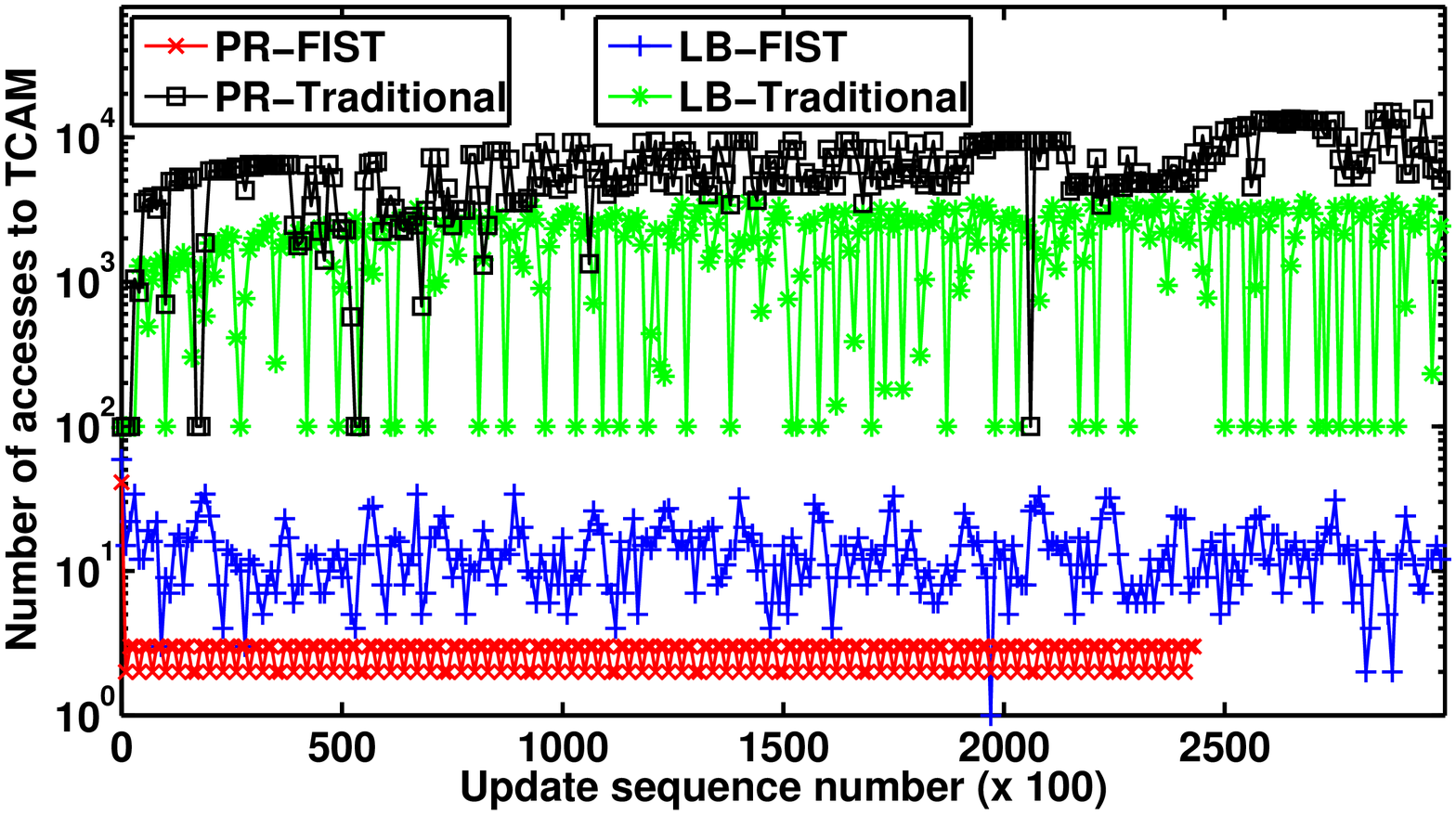}
\label{fig-access-tcam-frequency}%
}
\subfloat[Lookup speed with updates for policy routing]{
\centering
\includegraphics[width=2.1in]{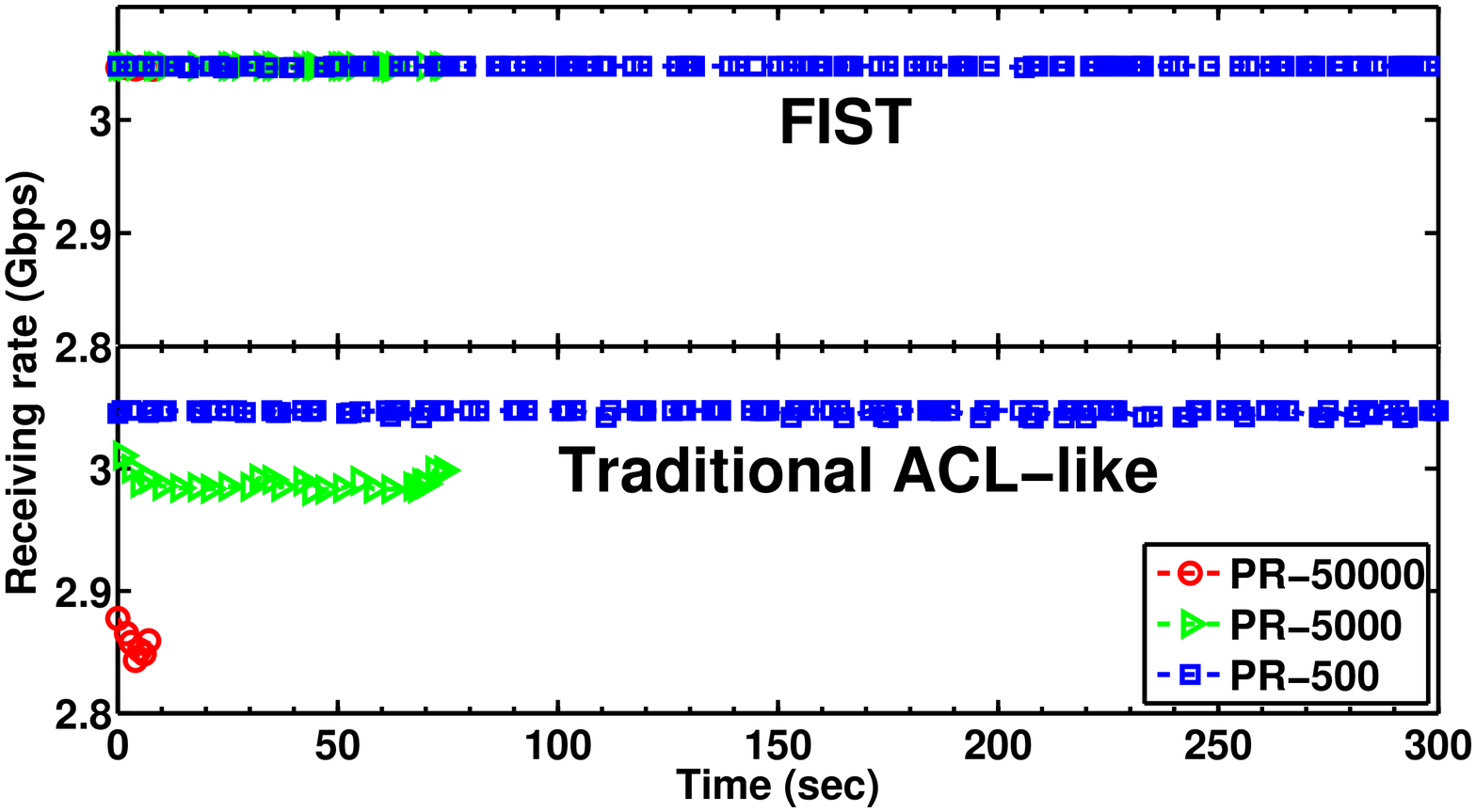}
\label{fig-evaluation-policy}%
}
\subfloat[Lookup speed with updates for load balancing]{
\centering
\includegraphics[width=2.1in]{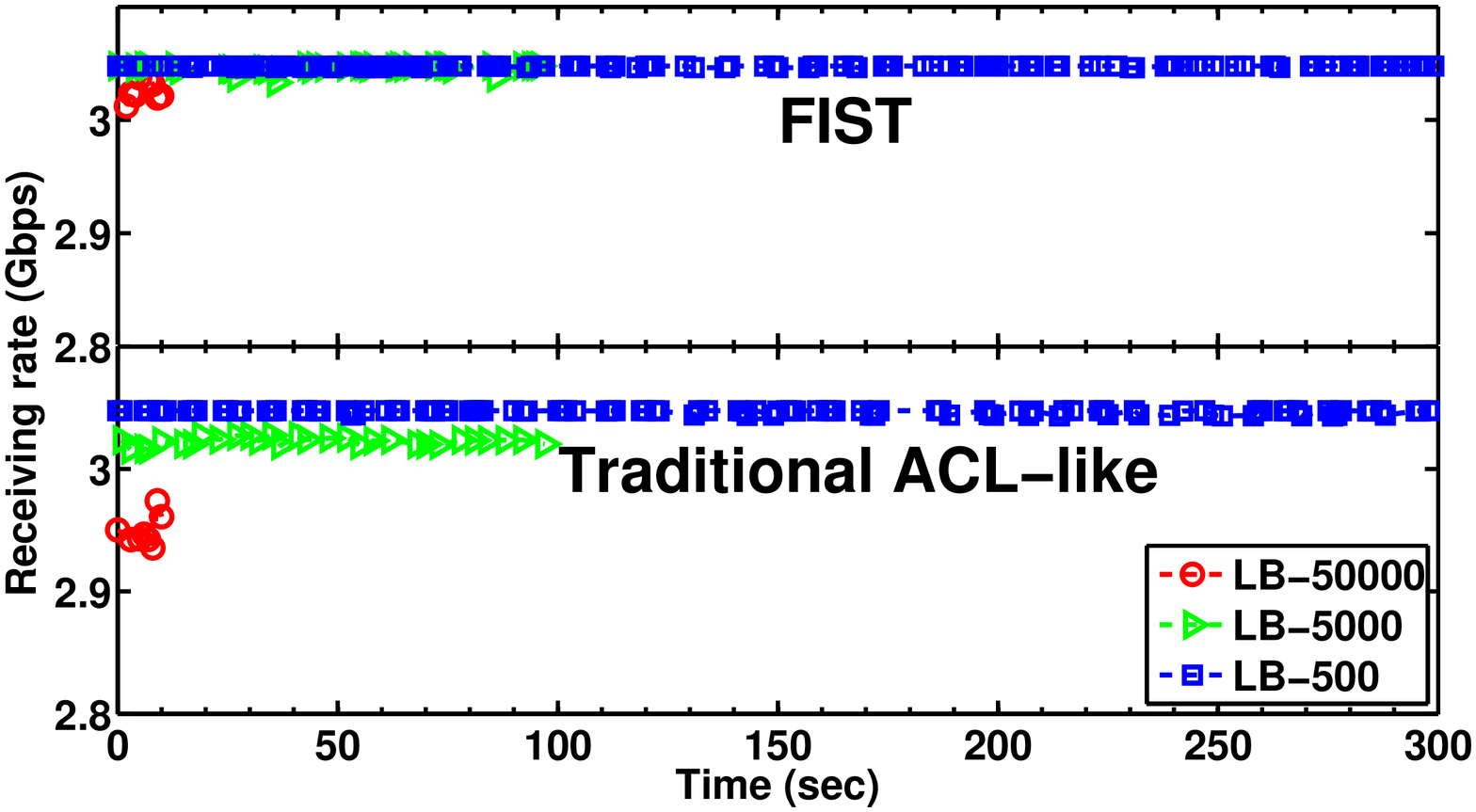}
\label{fig-evaluation-load}
}
\caption{\footnotesize{Lookup speeds with updates}}
\vspace{-0.3cm}
\end{figure*}

To evaluate the update burden, i.e., the influence of update on lookup speed for each update sequence. We first evaluate
the number of accesses to TCAM, as TCAM accesses dominates the interruption period during updates.  We also compare FIST
with ACL-like structure. In Figure \subref*{fig-access-tcam-frequency}, we show the number of accesses to TCAM
of FIST per 100 updates. We can see that PR brings only a few accesses to TCAM, this is because in our policy routing case,
all destination prefixes already exist in TCAM, and carrying out the policy routing only needs to assign rows to the destination
prefixes in destination table, and initially insert source prefixes into source table. After 243,500 updates, PR does not
need any access to TCAM, because all destination (or source) prefixes already exist in destination (or source) table. LB also
introduces only a few accesses to TCAM, because there exists many overlapping destination and source prefixes at different time
points, thus they do not have to be updated in destination and source tables each time. In contrast with FIST, ACL-like
structure introduces much more accesses to TCAM, e.g., 15,596 accesses to TCAM during 100 updates maximally. This is because 1) there are more
rules in forwarding table within ACL-like structure; 2) within FIST, we only have to guarantee the order of destination/source
prefixes with the same length in their respective destination/source table. However, within ACL-like structure, we have to
guarantee the order of (destination, source) prefix pairs with the same length (for both destination and source prefixes) in a common table.

In Figure \subref*{fig-evaluation-policy} and \subref*{fig-evaluation-load}, we show the lookup speed, i.e., receiving rate
on the traffic generator, of FIST within different update frequency during 5 minutes (when the update frequency is
5000, or 50000 updates/sec, the update process will be terminated earlier). In Figure \subref*{fig-evaluation-policy}, we can see
that within FIST structure, no matter at which frequency, updates has almost no influence on lookup in the policy routing scenario.
This is because PR cause only a few accesses to TCAM with each update, and brings a little interruption time during lookup.
In Figure \subref*{fig-evaluation-policy}, we also compare the results of FIST and ACL-like structures, we can see
that within ACL-like structure, updates have greater influence on lookup, e.g., the receiving rate is degrade by
about 7\% maximally when there are 50,000 updates per second.

In Figure \subref*{fig-evaluation-load}, we can see that within FIST structure, when the update frequency is low, i.e.,
500 updates per second. However, when the update frequency is high, e.g., 50,000 updates per second, the receiving rate
is degraded by about 2\%. This is because in our load balancing scenario, each updates cause more accesses to TCAM. Even
when update frequency is 5,000 updates per second, there still exists some time point when the lookup speed is degraded.
In Figure \subref*{fig-evaluation-load}, we can see that within ACL-like structure, even at the lowest frequency,
i.e., 500 updates per second, the performance is still degraded by about 0.1\%.

We conclude that our FIST structure will not introduce high update burden on lookup speed. In the policy routing
scenario, although there are may be millions of update when ISP operators decide to carry out some policies, the
update can be completed in a short time, e.g., less than 20 seconds when there are 1 million updates, without
having influence on lookup. Besides, in most cases, policy routing does not have to be implemented in a real-time
way. In the load balancing scenario, we perform updates every hour, we show the number of updates needed per hour
in Figure \ref{fig-load-balancing-updates}. The trend of updates per hour in Figure \ref{fig-load-balancing-updates}
is similar with the trend of traffic in CNGI-6IX in Figure \ref{fig_cernet2_utilization}, because we need to move
more traffic to CNGI-SHIX when CNGI-6IX is more congested. We can see that the maximum number of updates per hour
is about 1,300, which can be completed within one second without having influence on lookup.

\begin{figure}[!ht]%
\centering
\includegraphics[width=3.1in]{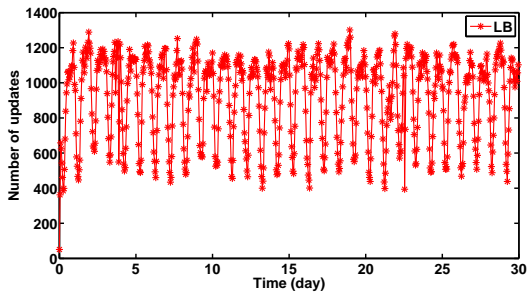}
\caption{\footnotesize{Number of updates for load balancing}}
\label{fig-load-balancing-updates}%
\end{figure}


\textbf{Number of Accesses to SRAM During Update:}
In Figure \ref{fig-comparison-saturation}, we show the number of accesses to SRAM within incremental update and
TD-Saturation(). We can see that for both policy routing and load balancing scenario, incremental update causes much less
accesses to SRAM. This is because during each update, TD-Saturation() has to reset all conflicted cells while incremental
update only has to reset the dependent cells that must be changed, which is a subset of all conflicted cells. For example,
in the load balancing scenario, incremental update causes 600 accesses to SRAM at most per 100 updates, while TD-Saturation
causes 10,814 accesses to SRAM at most per 100 updates. In the policy routing scenario, incremental update causes only 100
accesses to SRAM per 100 updates, this is because in the forwarding table of the policy routing scenario, the source prefixes
are composed of prefixes from two universities, i.e., THU and HUST. Prefixes from THU are totally disjoint, i.e., none source
prefix is a prefix of another, and prefixes from HUST are disjoint except for two prefixes (240c::/28 and 240c:3::/32).
Thus update a cell in TD-table will bring almost none conflicted cells.

In Figure \ref{fig-computation-time}, we also show the computation time per 100 updates for both incremental update and
TD-Saturation(). The result is similar with Figure \ref{fig-comparison-saturation}, because more accesses to SRAM indicates
more cells that have to be computed. Thus incremental update cost much less time per update, compared to TD-Saturation().

In Figure \ref{fig-isolation-sram}, we show the number of accesses to SRAM with and without isolating default next hop.
We only consider the load balancing scenario, because policy routing is a special case where all nodes in the colred tree
of any destination prefix are black, thus isolating default next hop has no effect. In the load balancing scenario,
we randomly insert 100 updates on the default next hops of destination prefixes, after each hour when load balancing is
carried out. In Figure \ref{fig-isolation-sram}, we can see that with isolation, each update on default next hop
bring none access to SRAM, because we only have to update in the TCAM. However, without isolation, each 100 updates
brings about 10,000 accesses to SRAM, because we also have to update the dependent cells in TD-table without isolation.

\begin{figure*}
\begin{minipage}[b]{0.66\linewidth}
\begin{minipage}[b]{0.49\linewidth}
\subfloat[Number of accesses to SRAM per 100 updates]{
\includegraphics[width=2.1in]{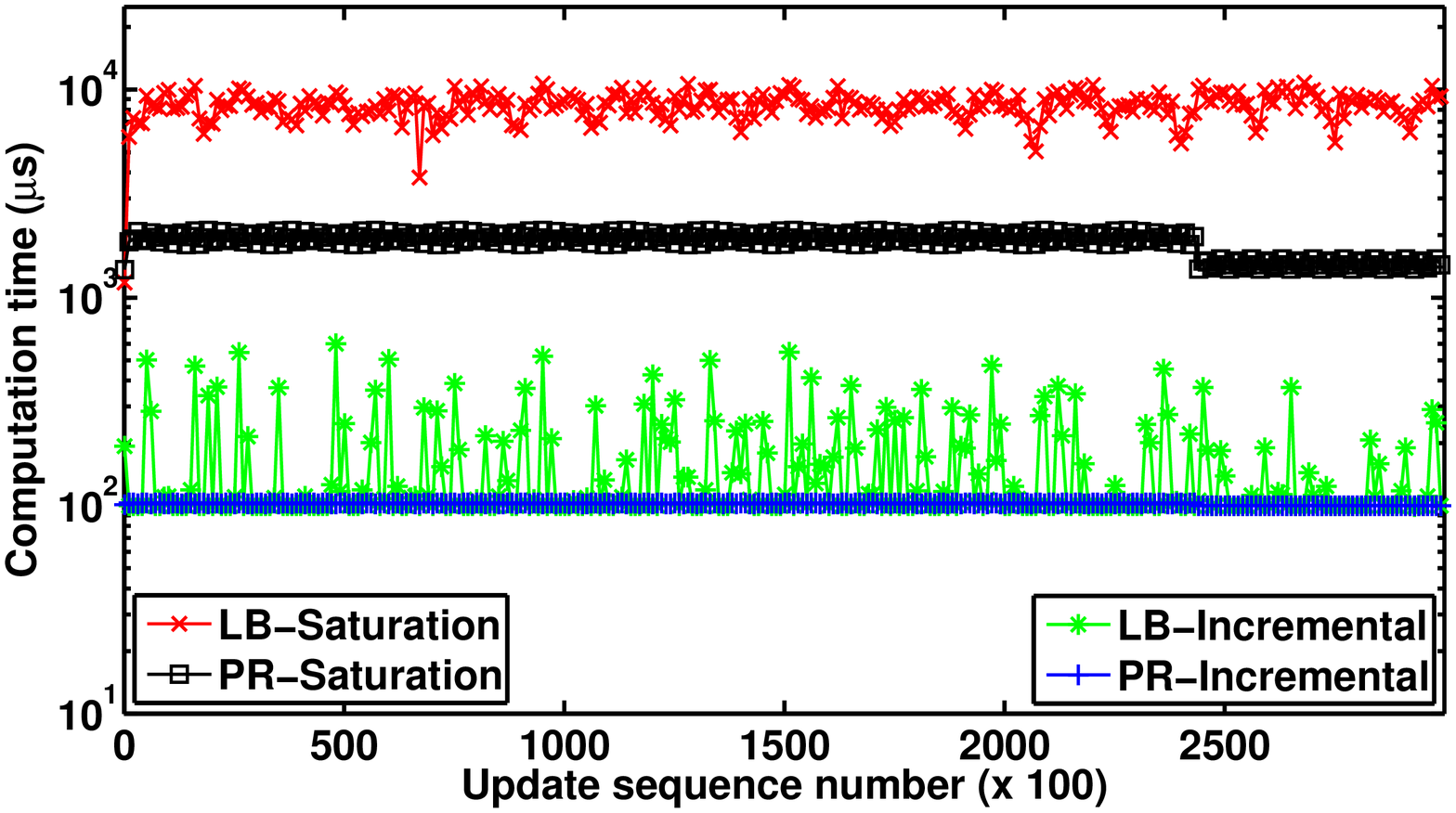}
\label{fig-access-sram-frequency}}
\end{minipage}
\begin{minipage}[b]{0.49\linewidth}
\subfloat[Computation time per 100 updates]{
\includegraphics[width=2.1in]{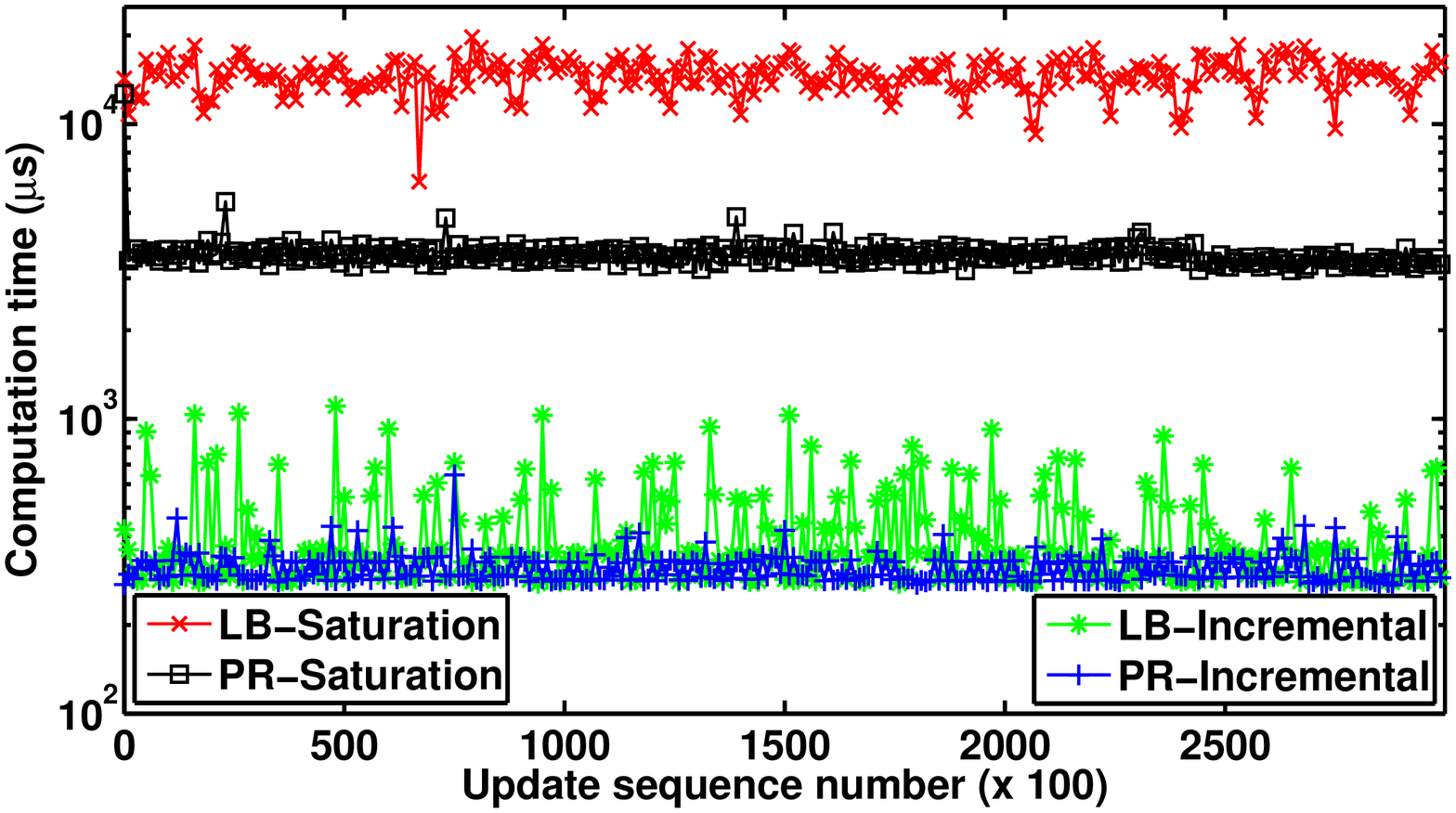}
\label{fig-computation-time}}
\end{minipage}
\caption{\footnotesize{Comparison between incremental updates and TD-Saturation()}}
\label{fig-comparison-saturation}
\end{minipage}
\begin{minipage}[b]{0.33\linewidth}
\begin{minipage}[b]{0.99\linewidth}
\includegraphics[width=2.1in]{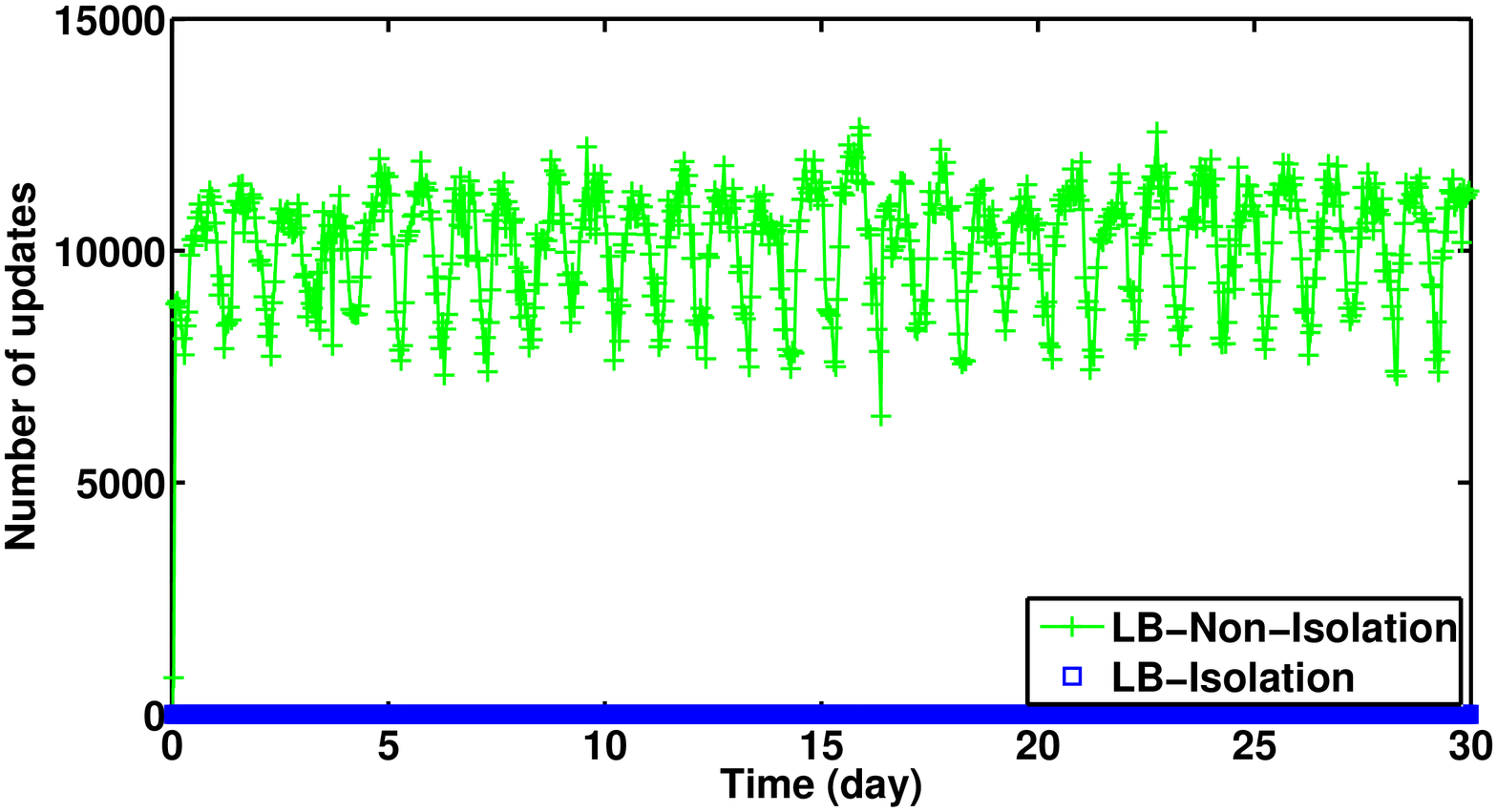}
\caption{\footnotesize{Comparison between isolation and non-isolation of default next hop}}
\label{fig-isolation-sram}
\end{minipage}
\end{minipage}
\end{figure*}

\section{Discussion about Scalability}\label{sec-discuss}
We admit that the trivial FIST will bring scalability issues in SRAM, if both the destination and source tables are very
large. Current largest SRAM chip in the market is 144Mb (288Mb SRAM is on the roadmap of major vendors) \cite{Hermsmeyer2009},
other memory products such as RLDRAM can provide similar performance (allows 16 bytes reading with random access time
of 15 ns with memory denominations of 576 Mbit/chip \cite{Fall06}. Suppose multiple chips (linecards of Bit-Engine 12004 support
four SRAM chips) is used, and 576Mb storage space is available for TD-table, if there are 10,000 destination prefixes,
then TD-table can accommodate at most 7550 source prefixes. It is obviously impractical within current 400,000 destination
prefixes.

However, the situation can be improved because 1) using non-homogenous structure can exclude most destination prefixes
from destination table; 2) in the real world, different prefixes usually share the same policy, e.g., prefixes belong
the the same university in CERNET2 should be equally treated. They can be compressed to the granularity of coarser granularity,
rather than prefixes; 3) we can enforce restrictions when adding a row or column into the TD-table. Beside, we are
making continuous efforts to eliminate the redundancies in TD-table.

update scalability

we admit that in certain circumstances,

\section{Related Work}\label{sec-related-work}
Packet classification is an important topic throughout the history of the Internet. With increasing demands from users
and ISPs for better and more flexible services, more research works focus on higher dimensional classification\cite{Chang09}
\cite{Qi09}. In layer-4, multi-dimensional classification is a familiar topic due to security and other reasons
\cite{Kim08}\cite{Lee11}. In layer-3, more and more routing schemes make routing decisions based on both source and
destination addresses, such as NIRA \cite{Yang07}, customer-specific routing \cite{Fu08}.
In this paper, our focus is on two dimensional classification in layer-3, i.e., designing a TwoD router.

\ignore{
Most of the previous works focus on software-based solutions, and various algorithmic structure have been devised to
reduce memory space, and speed up lookup and update. Such as tries (e.g., grid-of-tries, set-pruning tries
\cite{network-algorithmics}), (equivalent) cross-producting \cite{Wang09}, tuple space search \cite{Srinivasan99},
decision-tree based \cite{Vamanan10}, aggregated bit vector \cite{Baboescu05}, etc. However, due to their
non-deterministic characteristic, software-based solutions are not widely used in network routers.
}

Hardware-based, especially TCAM-based solutions are the de facto standard for the Internet routers \cite{meiners2010}.
TCAM-based solutions are limited by the capacity of TCAM \cite{Meiners11}, despite their constant lookup time. To reduce
the TCAM storage space, various compression schemes have been studied \cite{liu10}\cite{Meiners09}. In \cite{Suri03},
optimal two dimensional routing table compression is studied. Most works related with hardware-based multi-dimensional
classifiers are on the basic of traditional Cisco ACL structure, which is `fat' in TCAM and `thin' in SRAM.

In \cite{Meiners10-2},
a novel TCAM structure is proposed for firewall, it moves the majority information from expensive TCAM to cheaper SRAM.
However, it needs multiple sequential lookups in TCAM, and extending the width of TCAM entries, while TCAM chips storing
forwarding table have limited spare bits (in CERNET2, the TCAM width is set to be 144, only 16 bits are spared).
Such that it is not fit for our TwoD router design.

TCAM-based solutions need multiple accesses to memory during an update \cite{Luo12}. Nowadays, the update frequency can
reach tens of thousands per second \cite{Mishra10}, which seriously impedes the lookup speeds.
To solve this problem, \cite{Wang04}\cite{Mishra11} propose to keep the classification table lock-free, i.e., lookups will not
be interrupted by update. In this paper, we borrow their ideas during designing the update scheme.

\section{Conclusion}\label{sec-conclusion}
In this paper, we put forwarded a new forwarding table structure called FIST of TwoD routers, where forwarding decisions
is based on both destination and source addresses. Our focus is to accommodate the increasing number of rules in TwoD
routers, which is also a practical concern of CERNET2 during deploying TwoD-IP routing. Through making a novel separation
between TCAM and SRAM, FIST can significantly reduce the scarce TCAM storage space and keep fast lookup speed.

FIST stores destination and source prefixes in two separate TCAM tables. Combined with the matching results of both results,
we can find the next hop information for an arriving packet. Through pre-computation, we can resolve the potential confliction.
By proposing a new data structure called colered tree, we designed the incremental updating algorithm, that can minimize
the computation complexity and number of accesses to memory.

We implement the TwoD router within FIST on the linecard of a commercial router. Our design is incremental, and does not
need any new devices. We also made comprehensive with the real design and data sets from CERNET2. The results showed that
FIST can greatly reduce the TCAM storage space, and will not increase SRAM storage space in our scenarios.

\bibliographystyle{abbrv}
{\small{
\bibliography{tech}
}
}

\end{document}